\documentclass[a4paper,twoside,11pt,english,intlimits]{article}

\usepackage[utf8]{inputenc}
\usepackage[T2A,T1]{fontenc}
\DeclareSymbolFont{cyrillic}{T2A}{cmr}{m}{n}
\DeclareMathSymbol{\Sha}{\mathalpha}{cyrillic}{216}

\usepackage[english]{babel}

\usepackage{amsmath,amsthm,amssymb,stmaryrd}%\usepackage{upgreek}
\usepackage{mathtools}
\usepackage[ttscale=.875]{libertine}
\usepackage[libertine]{newtxmath}
\usepackage[commentmarkup=footnote]{changes}
\usepackage{setspace}

\usepackage{enumitem}%\usepackage{enumerate}
\usepackage{fancyhdr,titlesec,url}
\usepackage[all,knot,poly]{xy}
\usepackage{array}
\usepackage{hyperref}
\usepackage{graphicx}
\usepackage[numbers,sort&compress]{natbib}
\usepackage{multirow}
\usepackage{ifthen}
\usepackage{time}
\usepackage{algorithm2e}

\usepackage[capitalize,noabbrev]{cleveref}

\usepackage{etoolbox}

\usepackage{cancel}
\usepackage{tikz}
\usetikzlibrary{arrows.meta,calc,backgrounds}
\usepackage{tikz-cd}
\usetikzlibrary{matrix,decorations.markings,decorations.pathmorphing}
\usepackage{comment}
\usepackage{caption}

\newcounter{tempcounter}

\newcommand{\periodafter}[1]{\ifstrempty{#1}{}{#1.}}
\titleformat{\section}[block]{\scshape\filcenter\LARGE\boldmath}{\thesection.}{.5em}{}
\titleformat{\subsection}[block]{\bfseries\filcenter\large\boldmath}{\thesubsection.}{.5em}{\medskip}
\titleformat{\subsubsection}[runin]{\bfseries\boldmath}{\thesubsubsection.}{.5em}{\periodafter}%{}[.]
\titlespacing{\subsubsection}{0pt}{\topsep}{.5em}

\newtheoremstyle{ntheorem}%
	{\topsep}{\topsep}{\itshape}{0pt}{\bfseries}{.}{.5em}%
	{\thmnumber{#2.\hspace{.5em}}\thmname{#1}\thmnote{ (#3)}}
	
\newtheoremstyle{ndefinition}%
	{\topsep}{\topsep}{\normalfont}{0pt}{\bfseries}{.}{.5em}%
	{\thmnumber{#2.\hspace{.5em}}\thmname{#1}\thmnote{ (#3)}}
	
\newtheoremstyle{nremark}%
	{\topsep}{\topsep}{\normalfont}{0pt}{\itshape}{.}{.5em}%
	{\thmnumber{}\thmname{#1}\thmnote{ (#3)}}

\theoremstyle{ntheorem}
  	\newtheorem{theorem}[subsubsection]{Theorem}
  	\newtheorem{proposition}[subsubsection]{Proposition}
	\newtheorem{lemma}[subsubsection]{Lemma}

\theoremstyle{ndefinition}

	\newtheorem{remark}[subsubsection]{Remark}
	
\makeatletter
\def\@equationname{equation}
\newenvironment{eqn}[1]{%
    \def\mymathenvironmenttouse{#1}%
    \ifx\mymathenvironmenttouse\@equationname%
        \refstepcounter{subsubsection}%
    \else
        \patchcmd{\@arrayparboxrestore}{equation}{subsubsection}{}{}%          doesn't change output?
        \patchcmd{\print@eqnum}{equation}{subsubsection}{}{}%
        \patchcmd{\incr@eqnum}{equation}{subsubsection}{}{}%
    \fi
    \csname\mymathenvironmenttouse\endcsname%
}{%
    \ifx\mymathenvironmenttouse\@equationname%
        \tag{\thesubsubsection}%
    \fi
    \csname end\mymathenvironmenttouse\endcsname%
}
\makeatother

\fancypagestyle{plain}{
\fancyhf{}\fancyfoot[LC,RC]{}
\fancyfoot[LE,RO]{$\thepage$}

}

\UseTips
\SelectTips{eu}{11}

\newdir{ >}{{}*!/-10pt/@{>}}
\newdir{ -}{{}*!/-10pt/@{}}
\newdir{> }{{}*!/+10pt/@{>}}

\makeatletter

\xyletcsnamecsname@{dir4{}}{dir{}}
\xydefcsname@{dir4{-}}{\line@ \quadruple@\xydashh@}
\xydefcsname@{dir4{.}}{\point@ \quadruple@\xydashh@}
\xydefcsname@{dir4{~}}{\squiggle@ \quadruple@\xybsqlh@}
\xydefcsname@{dir4{>}}{\Tttip@}
\xydefcsname@{dir4{<}}{\reverseDirection@\Tttip@}

\xydef@\quadruple@#1{%
	\edef\Drop@@{%
		\dimen@=#1\relax
		\dimen@=.5\dimen@
		\A@=-\sinDirection\dimen@
		\B@=\cosDirection\dimen@
		\setboxz@h{%
			\setbox2=\hbox{\kern3\A@\raise3\B@\copy\z@}%
			\dp2=\z@ \ht2=\z@ \wd2=\z@ \box2
			\setbox2=\hbox{\kern\A@\raise\B@\copy\z@}%
			\dp2=\z@ \ht2=\z@ \wd2=\z@ \box2
			\setbox2=\hbox{\kern-\A@\raise-\B@\copy\z@}%
			\dp2=\z@ \ht2=\z@ \wd2=\z@ \box2
			\setbox2=\hbox{\kern-3\A@\raise-3\B@ \noexpand\boxz@}%
			\dp2=\z@ \ht2=\z@ \wd2=\z@ \box2
		}%
		\ht\z@=\z@ \dp\z@=\z@ \wd\z@=\z@ \noexpand\styledboxz@
	}%
}

\xydef@\Tttip@{\kern2pt \vrule height2pt depth2pt width\z@
	\Tttip@@ \kern2pt \egroup
	\U@c=0pt \D@c=0pt \L@c=0pt \R@c=0pt \Edge@c={\circleEdge}%
	\def\Leftness@{.5}\def\Upness@{.5}%
	\def\Drop@@{\styledboxz@}\def\Connect@@{\straight@{\dottedSpread@\jot}}}
	
\xydef@\Tttip@@{%
	\dimen@=.25\dimen@
 	\B@=\cosDirection\dimen@
	\setboxz@h\bgroup\reverseDirection@\line@ \wdz@=\z@ \ht\z@=\z@ \dp\z@=\z@
	{\vDirection@(1,-1)\xydashl@ \xyatipfont\char\DirectionChar}%
	{\vDirection@(1,+1)\xydashl@ \xybtipfont\char\DirectionChar}%
}

\xydef@\ar@form{
	\ifx \space@\next \expandafter\DN@\space{\xyFN@\ar@form}%
	\else\ifx ^\next \DN@ ^{\xyFN@\ar@style}\edef\arvariant@@{\string^}%
	\else\ifx _\next \DN@ _{\xyFN@\ar@style}\edef\arvariant@@{\string_}%
	\else\ifx 0\next \DN@ 0{\xyFN@\ar@style}\def\arvariant@@{0}%
	\else\ifx 1\next \DN@ 1{\xyFN@\ar@style}\def\arvariant@@{1}%
	\else\ifx 2\next \DN@ 2{\xyFN@\ar@style}\def\arvariant@@{2}%
	\else\ifx 3\next \DN@ 3{\xyFN@\ar@style}\def\arvariant@@{3}%
	\else\ifx 4\next \DN@ 4{\xyFN@\ar@style}\def\arvariant@@{4}%
	\else\ifx \bgroup\next \let\next@=\ar@style
	\else\ifx [\next \DN@[##1]{\ar@modifiers{[##1]}}%]
	\else\ifx *\next \DN@ *{\ar@modifiers}%
	\else\addLT@\ifx\next \let\next@=\ar@slide
	\else\ifx /\next \let\next@=\ar@curveslash
	\else\ifx (\next \let\next@=\ar@curveinout %)
	\else\addRQ@\ifx\next \addRQ@\DN@{\ar@curve@}%
	\else\addLQ@\ifx\next \addLQ@\DN@{\xyFN@\ar@curve}%
	\else\addDASH@\ifx\next \addDASH@\DN@{\defarstem@-\xyFN@\ar@}%
	\else\addEQ@\ifx\next \addEQ@\DN@{\def\arvariant@@{2}\defarstem@-\xyFN@\ar@}%
	\else\addDOT@\ifx\next \addDOT@\DN@{\defarstem@.\xyFN@\ar@}%
	\else\ifx :\next \DN@:{\def\arvariant@@{2}\defarstem@.\xyFN@\ar@}%
	\else\ifx ~\next \DN@~{\defarstem@~\xyFN@\ar@}%
	\else\ifx !\next \DN@!{\dasharstem@\xyFN@\ar@}%
	\else\ifx ?\next \DN@?{\ar@upsidedown\xyFN@\ar@}%
	\else \let\next@=\ar@error
	\fi\fi\fi\fi\fi\fi\fi\fi\fi\fi\fi\fi\fi\fi\fi\fi\fi\fi\fi\fi\fi\fi\fi \next@}

\makeatother

\newcommand{\fl}{\rightarrow}

\newcommand{\qfl}{\xymatrix@1@C=10pt{\ar@4 [r] &}}

\renewcommand{\tilde}[1]{\widetilde{#1}}
\newcommand{\tck}[1]{#1^{\top}}

\DeclareMathOperator{\id}{Id}

\renewcommand{\phi}{\varphi}
\renewcommand{\epsilon}{\varepsilon}

\newcommand{\Cr}{\mathcal{C}}

\newcommand{\Er}{\mathcal{E}}

\def\catego#1{\mathsf{#1}}

\DeclareMathOperator{\Nbb}{\mathbb{N}}

\DeclareMathOperator{\Ccal}{\mathcal{C}}
\DeclareMathOperator{\Dcal}{\mathcal{D}}

\newcommand{\Squa}[1]{\catego{Sq}_{#1}}
\newcommand{\SquaFill}[2]{\catego{SqF}_{#1}^{#2}}

\newcommand{\oto}[1]{\overset{#1}{\to}}
\newcommand{\ofrom}[1]{\overset{#1}{\leftarrow}}

\theoremstyle{ntheorem}

\tikzset{global scale/.style args={#1and#2}{scale=#1, every node/.append style={scale=#2}}}
\tikzcdset{global scale/.style args={#1and#2and#3and#4}{column sep={#1em,between origins}, row sep={#2em,between origins}, every label/.append style={scale=#4}, cells={nodes={scale=#3}}, nodes in empty cells}}
\tikzcdset{longer arrows/.style args={#1and#2}{every arrow/.append style={shorten <= -#1, shorten >= -#2}}}
\tikzset{triangle/.style = {path picture={\draw (path picture bounding box.north west) -- (path picture bounding box.north east) -- (path picture bounding box.south) -- (path picture bounding box.north west);}}}
\tikzset{Rightarrow/.style={double equal sign distance,>={Implies},->}, Rrightarrow/.style={-,preaction={draw,Rightarrow}}, Rrrightarrow/.style={preaction={draw,Rightarrow},-,double,double distance=0.2pt}}

\tikzcdset{arrow style=tikz, diagrams={>=stealth}}

\makeatletter
\providecommand{\leftsquigarrow}{%
  \mathrel{\mathpalette\reflect@squig\relax}%
}
\newcommand{\reflect@squig}[2]{%
  \reflectbox{$\m@th#1\rightsquigarrow$}%
}
\makeatother

\DeclareFontFamily{U}{mathx}{\hyphenchar\font45}
\DeclareFontShape{U}{mathx}{m}{n}{
      <5> <6> <7> <8> <9> <10>
      <10.95> <12> <14.4> <17.28> <20.74> <24.88>
      mathx10
      }{}
\DeclareSymbolFont{mathx}{U}{mathx}{m}{n}
\DeclareFontSubstitution{U}{mathx}{m}{n}
\DeclareMathAccent{\widecheck}{0}{mathx}{"71}

\newcommand{\cat}[1]{#1^{\ast}}

\newcommand{\cstabs}[1]{#1^{\circ_i}}
\newcommand{\cstabsr}[1]{#1^{\top_i}}%\newcommand{\cstabsr}[1]{#1^{\circ_i,R_i}}

\newcommand{\onebf}{{\bf 1}}

\renewcommand{\id}{id}

\renewcommand{\leq}{\leqslant}
\renewcommand{\geq}{\geqslant}

\newcommand{\SSS}[1]{§\ref{#1}}

\newcount\hh
\newcount\mm
\mm=\time
\hh=\time
\divide\hh by 60
\divide\mm by 60
\multiply\mm by 60
\mm=-\mm
\advance\mm by \time
\def\hhmm{\number\hh:\ifnum\mm<10{}0\fi\number\mm}
\definecolor{vert}{rgb}{0,0.45,0}
\definecolor{rouge}{rgb}{0.89,0.04,0.36}
\definecolor{MyGray}{gray}{0.6}
\definecolor{MyRed}{RGB}{212,42,42}
\newcommand{\cb}[3]{\mathrm{C}^{#1}_{#2}(#3)}
\newcommand{\res}[2]{#1 \mathbin{|} #2}

\newcommand{\Abst}[1]{{\mathsf{#1}}}
\newcommand{\Brc}[2]{\Abst{B}_{#1}(#2)}
\newcommand{\LBrc}[2]{\Abst{LB}_{#1}(#2)}
\newcommand{\Cf}[2]{\Abst{Cf}_{#1}(#2)}
\newcommand{\LCf}[2]{\Abst{LCf}_{#1}(#2)}
\newcommand{\NCf}[2]{\Abst{NCf}_{#1}(#2)}
\newcommand{\NLCf}[2]{\Abst{NLCf}_{#1}(#2)}
\newcommand{\CR}[1]{\Abst{CR}(#1)}
\newcommand{\NCR}[1]{\Abst{NCR}(#1)}

\def\catego#1{\mathsf{#1}}
\newcommand{\extGamma}[3]{\Gamma_{({#1} ~ {#2})}^{#3}}

\newcommand{\PreCub}[1]{{\catego{PreCub}_{#1}}}			%cubique
\newcommand{\GlobCat}[1]{{\catego{Cat}_{#1}}}
\newcommand{\CubCat}[1]{{\catego{Cub}_{#1}}}				%cubique
\newcommand{\CubCatG}[1]{{\CubCat{#1}^{\Gamma}}}			%cubique

\newcommand{\CubPol}[1]{{\catego{CubPol}_{#1}}}			%cubique

\DeclareMathOperator{\Xcal}{\mathcal{X}}

\newcommand{\filler}[2]{{A_{#1}(#2)}}		% confluence filler
\newcommand{\fillerCR}[1]{{B(#1)}}	% Church-Rosser filler

\newcommand{\auteur}[3]{
\noindent
\begin{minipage}[t]{.45\textwidth}
\begin{flushright}
\textsc{#1} \\
{\footnotesize\textsf{#2}}
\end{flushright} 
\end{minipage}
\qquad
\begin{minipage}[t]{.45\textwidth}
#3
\end{minipage}
}

\begin{document}
\thispagestyle{empty}

\begin{center}

% Titre
\begin{doublespace}
\begin{huge}
{\scshape Cubical coherent confluence,}
\end{huge}

\vskip+1.5pt

\begin{huge}
{\scshape $\omega$-groupoids and the cube equation}
\end{huge}

\vskip+2pt

\bigskip
\hrule height 1.5pt 
\bigskip

\vskip+5pt

% Auteurs
\begin{Large}
{\scshape Philippe Malbos - Tanguy Massacrier - Georg Struth}
\end{Large}
\end{doublespace}

%%%
%\vfill

%%%
\vskip+20pt

\begin{small}\begin{minipage}{14cm}
\noindent\textbf{Abstract --}
We study the confluence property of abstract rewriting systems
internal to cubical categories. We introduce cubical contractions, a
higher-dimensional generalisation of reductions to normal forms, and
employ them to construct cubical polygraphic resolutions of convergent
rewriting systems. Within this categorical framework, we establish
cubical proofs of fundamental rewriting results -- Newman’s lemma, the
Church–Rosser theorem, and Squier’s coherence theorem -- via the
pasting of cubical coherence cells. We moreover derive, in purely
categorical terms, the cube law known from the $\lambda$-calculus and
Garside theory. As a consequence, we show that every convergent
abstract rewriting system freely generates an acyclic cubical
groupoid, in which higher-dimensional generators can be replaced by
degenerate cells beyond dimension two.

\medskip


\smallskip\noindent\textbf{M.S.C. 2020 --} 03B35, 68Q42, 18N30.
\end{minipage}
\end{small}
\end{center}

%\vskip+0pt

\begin{center}
%%% Table des matieres
\begin{small}\begin{minipage}{12cm}
\renewcommand{\contentsname}{}
\setcounter{tocdepth}{1}
\tableofcontents
\end{minipage}
\end{small}
\end{center}

\vskip+10pt

\section{Introduction}

This work started from the study of $n$-branchings of rewriting paths
in polygraphic resolutions and homotopical reduction-completion
procedures of higher-dimensional rewriting
systems~\cite{GuiraudMalbos12advances,GaussentGuiraudMalbos15}. Such
branchings can be regarded as computations starting in the same
state. An important property of branching computations is confluence,
which holds if these computations may eventually join in a common
state. Higher-dimensional rewriting is usually based on strict
$\omega$-categories~\cite{Polybook2025}, wich compose cells of
globular shape. Yet it often seems more natural to assemble confluence
and other rewriting diagrams into higher-dimensional cubes. So why not
use cubical categories instead for rewriting?

The relationship between rewriting theory~\cite{Terese03} -- a
fundamental model of computation with far-reaching applications in
mathematics and computer science -- and higher globular categories is
natural and well studied~\cite{Polybook2025}. We consider it in its
purest form through abstract rewriting systems, through ($1$-poly)graphs
$\partial^-,\partial^+:X_1\to X_0$, where $X_0$ is a set of $0$-cells
or vertices, $X_1$ is a set of $1$-cells or directed edges, and
$\partial^-$, $\partial^+$ are source and target maps relating
them. A rewriting path or computation is then a morphism or $1$-cell
in the (free) path category generated by such a graph. Higher
structure emerges in rewriting either through structured objects, or
alternatively through relationships between rewriting paths and higher
relationships between higher relationships. The free monoid used in
string rewriting, for instance, is a category with a single $0$-cell;
rewriting steps then become $2$-cells. Alternatively, in the left
square below, the $2$-cell $A$ expresses a relationship between the
rewriting paths along its faces.
\begin{eqn}{equation*}
\begin{tikzcd}[global scale = 2 and 2 and 1 and 1.2]
  w \ar[rr, "f"] \ar[dd, "g"'] && x \ar[dd, "h"] \\
&A&\\
y \ar[rr, "k"'] && z
\end{tikzcd}
\qquad\qquad\qquad
\begin{tikzcd}[global scale = 2 and 2 and 1 and 1.2]
  w \ar[rr, "f"] \ar[dd, "g"'] && x \ar[dd, "\res{g}{f}"] \\
  &A(f,g)&\\
y \ar[rr, "\res{f}{g}"'] && z
\end{tikzcd}
\end{eqn}
The square on the right expresses confluence of the branching
$y\xleftarrow{g} w\xrightarrow{f} x$ more specifically in the sense
that the paths $f$, $g$ can be extended from $y$ and $z$ to some
common vertex $z$, the notation $A(f,g)$ indicating the existential
dependency of its faces $f\vert g$ and $g\vert f$ on $f$ and
$g$. Likewise, confluences of $n$-branchings lead naturally to
coherence $n$-cubes, which globular categories obviously model as
globes.

Rewriting with higher cells requires higher-dimensional rewriting
systems supplying generators, relations and rewriting paths in higher
dimensions: so-called computads~\cite{Street76,Street87} or
polygraphs~\cite{Burroni93}.  Polygraphic
resolutions~\cite{Metayer03,GuiraudMalbos12advances,Polybook2025} then
amount to the construction of higher-dimensional rewriting systems
with desirable properties such as confluence and termination
guarantees. When rewriting with structured objects, these can be
obtained via reduction-completion procedures that resolve obstacles to
confluence given by certain
$n$-branchings~\cite{GaussentGuiraudMalbos15}.  These have been
developed for resolving algebraic and categorical structures in
homological algebra for
categories~\cite{Metayer03,GuiraudMalbos12advances}, associative
algebras~\cite{GuiraudHoffbeckMalbos19,LiuMalbos2025} and
operads~\cite{MalbosRen23}, as well as for
algebraic~\cite{GaussentGuiraudMalbos15} and
categorical~\cite{CurienMimram17} coherence proofs.

Proofs about rewriting systems are often presented in semi-formal
diagrammatic style. The literature abounds in particular with diagrams
gluing cubes~\cite{Terese03,Barendregt84}. In higher-dimensional
rewriting, this amounts to composing higher cells in the underlying
categories.

The idea of using cubical categories for higher-dimensional rewriting
is not new. A cubical approach has been pioneered by
Lucas~\cite{LucasPhD2017,Lucas2018,Lucas20}, building on Brown and
Higgin's cubical categories~\cite{BrownH81,AlAglBrownSteiner2002},
which in turn add compositions to the cubical sets of
Serre~\cite{Serre1951} and Kan~\cite{Kan1955}.  Lucas has in
particular proved the existence of cubical polygraphs, adapting ideas
by Batanin~\cite{Batanin98} and Garner~\cite{Garner2010}. His
polygraphs carry a monoidal structure to capture ``string'' rewriting
with monoid objects. Using this formalism he has verified some
standard confluence properties using cubical $2$-polygraphs, and
studied certain polygraphic resolutions for monoids. Our work is
strongly influenced by his. Al-Agl, Brown and Steiner have
shown that cubical categories with connection maps are equivalent to
globular ones~\cite{AlAglBrownSteiner2002}, which suggests that one
may translate between these two approaches to higher-dimensional
rewriting.

Higher confluence properties, in dimension $3$ and with emphasis on
cubes, have received longstanding interest in the rewriting
literature, too. Lévy has derived a cube law in the
$\lambda$-calculus, showing that all $3$-branchings of certain
rewriting paths of $\lambda$-terms extend around the edges of
$3$-dimensional confluence cubes~\cite{LevyPhD78}. Several sections in
Barendregt's monograph on the $\lambda$-calculus~\cite{Barendregt84}
are devoted to this cube law and a theory of residuals akin to
$f\vert g$ and $g\vert f$ in the diagram above. A comprehensive survey
on the cube law in rewriting has been writen by Endrullis and
Klop~\cite{EndrullisKlop2019}, including work by Klop himself, who has
returned to $3$-confluences and the cube law several times within four
decades. Endrullis and Klop not only open up fascinating relationships
with knot and Garside theory~\cite{EndrullisKlop2019,Dehornoy2015},
they also use the cube law as a hypothesis for a $3$-confluence
proof. By contrast, van Oostrom has recently sketched a combinatorial
bricklaying procedure for $3$-confluence proofs that is meant to
satisfy the cube law by construction~\cite{vanOostrom23}.

Here, we combine the two lines of work on cubical higher-dimensional
rewriting and higher confluence proofs in the context of polygraphic
resolutions of higher-dimensional cubical abstract rewriting systems,
which we present as constructions of certain cubical
$\omega$-groupoids.

To this end, we first extend the framework of cubical
higher-dimensional rewriting with contractions, which are essential
for constructing cubical polygraphs with the rewriting properties
desired. For this, we work with cubical $(\omega,p)$-categories where
cells in dimensions greater than $p+1$ are invertible. Their
definitions are recalled in Section~\ref{S:CubicalCategory}. Our
notion of contraction, introduced in Section~\ref{S:CubContrStrat}, is
given by a family of lax
transformations~\cite{AlAglBrownSteiner2002,LucasPhD2017,Lucas2018}, a
generalisation of natural transformations to cubical
categories. Intuitively, contractions extend rewriting strategies to
higher dimensions. For their definition, we first impose a quotient
structure in dimension $p$ on the underlying $(\omega,p)$-category,
and then define a section as a choice of a representative, for
instance a normal form.  Contractions extend this choice function
recursively to higher dimensions. This leads to a notion of
contracting cubical $(\omega,p)$-category, in which all cells of
dimension greater than $p+1$ can be contracted. The main result in
this context, Theorem~\ref{T:ContractingImpliesAcyclic}, shows that
every contracting $(\omega,0)$-category (hence every cubical
$\omega$-groupoid) is acyclic, so that all boundaries with a cubical
hole can be filled with a cell.

As examples of abstract cubical rewriting, we revisit some classical
diagrammatic confluence proofs in higher dimensions as cubical cell
compositions in Section~\ref{S:CCConfluence} , including variants of
Newman's lemma and the Church-Rosser theorem in two cubical
directions.  We also prove a variant of Squier's
theorem~\cite{SquierOttoKobayashi94}, which requires contractions and
can be seen as a low-dimensional version of
Theorem~\ref{T:ContractingImpliesAcyclic} for confluent and
terminating rewriting systems. In particular, we present a proof of
Newman's lemma in three cubical directions without explicitly use of
the cube law, as it is an immediate consequence of the geometry
imposed by the axioms of cubical categories. Using contractions, we
can even derive the cube law without involving coherence $3$-cells. To
simplify proofs, we use an internal abstract rewriting system in an
$(\omega,p)$-category, which can be seen as a generalisation of a
polygraph.

Our final contribution, in Section~\ref{S:PolygraphicResolutionARS},
lies in the study of polygraphic resolutions of cubical categories.
More specifically, we construct an acyclic cubical $\omega$-groupoid
from an abstract rewrite systems $\partial^-,\partial^+:X_1\to X_0$,
using a normalisation strategy based on contractions. For this, we
first introduce an explicit construction of cubical polygraphs and
prove Theorem~\ref{T:AcyclicityNormalisation}, a converse of
Theorem~\ref{T:ContractingImpliesAcyclic}, showing that free cubical
$\omega$-groupoids on polygraphs are acyclic if and only they are
contracting. We then turn to polygraphic resolutions of confluent and
terminating abstract rewriting systems, extending them recursively to
acyclic $\omega$-groupoids in Theorem~\ref{T:AcyclicExtensionARS},
which involves studying their $n$-branchings. Finally, in
Theorem~\ref{T:TruncatedAcyclicExtensionARS}, we refine this
construction so that it generates no non-trivial higher cells in
dimension greater than $2$. This result confirms in a more structural
way that the cube law does not require coherence $3$-cells in our
setting.  For abstract rewriting systems, no cubes are needed, because
homotopically, all cubes are empty.

In combination, these contributions shed in particular some light on
the cube law and address a longstanding question in the rewriting
community, which has been asked quite poignantly by
Klop~\cite{Klop2023}: ``\emph{One would expect [...] in higher
  category theory [...] that the Cube Equation [...] would be very
  much present [...]. But it seems that the contrary is the case:
  nowhere [...] one encounters the Cube Equation or residual
  notions. (I would love to be corrected!) How come? [...] Is a
  fundamental notion as \textbf{confluence} a total stranger in
  categories? }''.

\section{Preliminaries on Cubical Categories}
\label{S:CubicalCategory}

Cubical categories, introduced by Brown and
Higgins~\cite{BrownHiggins1981,BrownHiggins1981b}, are cubical sets
equipped with partial composition operations along the faces of
higher-dimensional cubes, and with identity cells in every dimension.
In this section, we adopt the axioms of Al-Agl, Brown and
Steiner~\cite{AlAglBrownSteiner2002}, augmented with the
cell-invertibility structure introduced by Lucas~\cite{LucasPhD2017},
and we recall the notion of lax transformations of cubical
categories—referred to as $1$-fold left homotopies
in~\cite{AlAglBrownSteiner2002}. Our setting is that of cubical
$\omega$-categories, possibly equipped with connections and inverses,
as formalised in~\cite{GrandisMauri2003}. For each $n\in \mathbb{N}$,
a cubical $n$-category is defined as the truncation of a cubical
$\omega$-category.

\subsection{Cubical $\omega$-categories}
\label{SS:CubicalCategories}

We henceforth assume that Greek letters $\alpha,\beta$ occurring as
superscripts of operators range over~$\{-,+\}$.

\subsubsection{}
\label{SSS:DefCubCat}
A \emph{cubical $\omega$-category} $\Ccal$ consists of
\begin{enumerate}
\item\label{I:CubicalCells} a family $(\Ccal_k)_{0\leq k}$ of sets of \emph{$k$-cells of $\Ccal$},
\item\label{I:CubicalRelations} \emph{face maps} $\partial_{k,i}^\alpha : \Ccal_k \to \Ccal_{k-1}$, for $1\leq i\leq k$, satisfying the \emph{cubical relations} 
\begin{eqn}{equation}
\label{E:AxiomPreCubClass}
\partial_{k-1,i}^\alpha\partial_{k,j}^\beta=\partial_{k-1,j-1}^\beta\partial_{k,i}^\alpha\qquad
(1\le i<j<k),
\end{eqn}
\item \emph{degeneracy maps} $\epsilon_{k,i}:\Ccal_{k-1}\to\Ccal_k$, for $1\leq i\leq k$,
\item \emph{composition maps} $\circ_{k,i}:\Ccal_k\times_{k,i}\Ccal_k \fl \Ccal_k$, for $1\leq i\leq k$, defined on the pullback $\Ccal_k\times_{k,i}\Ccal_k$ of the cospan $\Ccal_k\oto{\partial_{k,i}^+}\Ccal_{k-1}\ofrom{\partial_{k,i}^-}\Ccal_k$.
\setcounter{tempcounter}{\value{enumi}}
\end{enumerate}
These data are subject to the relations listed in
Appendix~\ref{AA:AxiomsCubCat}. Throughout this paper we consider cubical $\omega$-categories with
\begin{enumerate}
\setcounter{enumi}{\value{tempcounter}}
\item \emph{connection maps} $\Gamma_{k,i}^\alpha:\Ccal_{k-1}\to\Ccal_k$, for
$1\leq i<k$, satisfying the relations in Appendix~\ref{AA:AxiomsCubCatConnections}.
\end{enumerate}

A \emph{functor} $F: \Ccal \to \Dcal$ \emph{of cubical $\omega$-categories} is a
family of maps $(F_k:\Ccal_k\to\Dcal_k)_{0\leq k}$ that
preserve all face, degeneracy, composition and connection maps, see
Appendix~\ref{AA:CubicalFunctors}.  

All categories considered are cubical, so we drop this adjective wherever possible.

\subsubsection{}
\label{SSS:IllustratingCubCat}
Any $k$-cell $A$ and its faces can be represented, for
$1 \leq i < j \leq k$, by the diagram
\begin{equation*}
\begin{tikzcd}[global scale = 2 and 2 and 1 and 1.2]
\ar[r, shorten <= -5] \ar[d, shorten <= -7] & i \\
j & 
\end{tikzcd}
\qquad
\begin{tikzcd}[global scale = 4 and 2.1 and 1 and 1.2]
\partial_{k-1,i}^- \partial_{k,j}^- A \ar[rr, "\partial_{k,i}^- A"] \ar[dd, "\partial_{k,j}^- A"'] && \partial_{k-1,i}^- \partial_{k,j}^+ A \ar[dd, "\partial_{k,j}^+ A"] \\
 & A & \\
\partial_{k-1,i}^+ \partial_{k,j}^- A \ar[rr, "\partial_{k,i}^+ A"'] && \partial_{k-1,i}^+ \partial_{k,j}^+ A
\end{tikzcd}
\end{equation*}
The arrows on the left indicate the two directions along which the
faces of the cell $A$ are drawn.
\emph{Degeneracies}, cells in the codomains of degeneracy maps, are
illustrated as follows, where boxes as those on the right have been
introduced in~\cite{AlAglBrownSteiner2002}:
\begin{equation*}
\begin{tikzcd}[global scale = 2 and 2 and 1 and 1.2]
\ar[r, shorten <= -5] \ar[d, shorten <= -7] & i \\
j & 
\end{tikzcd}
\qquad
\begin{tikzcd}[global scale = 2 and 2 and 1 and 1.2]
x \ar[rr, equal] \ar[dd, "f"'] && x \ar[dd, "f"] \\
 & \epsilon_{k,i} f & \\
y \ar[rr, equal] && y
\end{tikzcd}
\quad
\text{or}
\quad
\vcenter{\hbox{
\begin{tikzpicture}[global scale = 1 and 1]
\draw [-] (0,0) -- (1,0);
\draw [-] (0,1) -- (1,1);
\draw [-] (0,0) -- (0,1);
\draw [-] (1,0) -- (1,1);
\draw [-] (0.2,0.5) -- (0.8,0.5);
\end{tikzpicture}
}}
\qquad\qquad
\begin{tikzcd}[global scale = 2 and 2 and 1 and 1.2]
x \ar[rr, "f"] \ar[dd, equal] && y \ar[dd, equal] \\
 & \epsilon_{k,j} f & \\
x \ar[rr, "f"'] && y
\end{tikzcd}
\quad
\text{or}
\quad
\vcenter{\hbox{
\begin{tikzpicture}[global scale = 1 and 1]
\draw [-] (0,0) -- (1,0);
\draw [-] (0,1) -- (1,1);
\draw [-] (0,0) -- (0,1);
\draw [-] (1,0) -- (1,1);
\draw [-] (0.5,0.2) -- (0.5,0.8);
\end{tikzpicture}
}}
\end{equation*}
The arrows between the two copies of $x$ or $y$ are drawn as equality
arrows to indicate that these faces are themselves degenerate.

The $\circ_{k,i}$-composition of two $k$-cells $A$, $B$ in
direction $i$ glues these cells along $i$ if the upper faces of
the first cell in all other directions match the lower faces in all other
directions of the second:
\begin{equation*}
\begin{tikzcd}[global scale = 2 and 2 and 1 and 1.2]
\ar[r, shorten <= -5] \ar[d, shorten <= -7] & i \\
j & 
\end{tikzcd}
\qquad
\begin{tikzcd}[global scale = 2 and 2 and 1 and 1.2]
\ar[rr] \ar[dd, shorten <= -4] && \ar[dd, shorten <= -4, dotted, "f"] && \ar[rr] \ar[dd, shorten <= -4, dotted, "f"'] && \ar[dd, shorten <= -4] && \ar[rr] \ar[dd, shorten <= -4] && \ar[dd, shorten <= -4] \\
 & A && \circ_{k,i} && B && = && A\circ_{k,i}B & \\
\ar[rr] &&&& \ar[rr] &&&& \ar[rr] && 
\end{tikzcd}
\qquad
\text{or}
\qquad
\vcenter{\hbox{
\begin{tikzpicture}[global scale = 1 and 1]
\node [] () at (0.5,0.5) {$A$};
\node [] () at (1.5,0.5) {$B$};
\draw [-] (0,0) -- (2,0);
\draw [-] (0,1) -- (2,1);
\draw [-] (0,0) -- (0,1);
\draw [-] (1,0) -- (1,1);
\draw [-] (2,0) -- (2,1);
\end{tikzpicture}
}}
\end{equation*}
Such diagrams make it easy to check that the degeneracies
$\epsilon_{k,i}$ provide identities for the $\circ_{k,i}$-composition.

\emph{Connections} are cells in the codomains of the
connection maps~$\Gamma_{k,i}$. Their diagrams are as
follows~\cite{AlAglBrownSteiner2002}:
\begin{equation*}
\begin{tikzcd}[global scale = 2 and 2 and 1 and 1.2]
\ar[r, shorten <= -5] \ar[d, shorten <= -7] & i \\
j & 
\end{tikzcd}
\qquad
\begin{tikzcd}[global scale = 2 and 2 and 1 and 1.2]
x \ar[rr, "f"] \ar[dd, "f"'] && y \ar[dd, equal] \\
 & \Gamma_{k,i}^- f & \\
y \ar[rr, equal] && y
\end{tikzcd}
\quad
\text{or}
\quad
\vcenter{\hbox{
\begin{tikzpicture}[global scale = 1 and 1]
\draw [-] (0,0) -- (1,0);
\draw [-] (0,1) -- (1,1);
\draw [-] (0,0) -- (0,1);
\draw [-] (1,0) -- (1,1);
\draw [-] (0.2,0.5) -- (0.5,0.5) -- (0.5,0.8);
\end{tikzpicture}
}}
\qquad\qquad
\begin{tikzcd}[global scale = 2 and 2 and 1 and 1.2]
x \ar[rr, equal] \ar[dd, equal] && x \ar[dd, "f"] \\
 & \Gamma_{k,i}^+ f & \\
x \ar[rr, "f"'] && y
\end{tikzcd}
\quad
\text{or}
\quad
\vcenter{\hbox{
\begin{tikzpicture}[global scale = 1 and 1]
\draw [-] (0,0) -- (1,0);
\draw [-] (0,1) -- (1,1);
\draw [-] (0,0) -- (0,1);
\draw [-] (1,0) -- (1,1);
\draw [-] (0.5,0.2) -- (0.5,0.5) -- (0.8,0.5);
\end{tikzpicture}
}}
\end{equation*}

A cell in $\Ccal$ is \emph{thin} if it is a composite of degeneracies
and
connections~\cite{BrownHiggins1977,BrownHiggins1981,BrownHiggins1981b}. An
example is
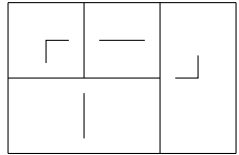
\begin{equation*}
\begin{tikzpicture}[global scale = 1 and 1]
\draw [-] (0,0) -- (0,2);
\draw [-] (0,0) -- (3,0);
\draw [-] (0,1) -- (2,1);
\draw [-] (0,2) -- (3,2);
\draw [-] (1,1) -- (1,2);
\draw [-] (2,0) -- (2,2);
\draw [-] (3,0) -- (3,2);
\draw [-] (0.5,1.2) -- (0.5,1.5) -- (0.8,1.5);
\draw [-] (1,0.2) -- (1,0.8);
\draw [-] (1.2,1.5) -- (1.8,1.5);
\draw [-] (2.2,1) -- (2.5,1) -- (2.5,1.3);
\end{tikzpicture}
\end{equation*}
We follow common practice and omit dimension indices $k$ if suitable.

\subsection{Cubical $(\omega,p)$-categories and lax transformations}
\label{SS:Cubical_omega_p_LaxTransformations}

\subsubsection{Invertibility}
\label{SS:DefGrpdCubCat}

Invertible cubical cells were introduced by Brown and
Higgins~\cite{BrownHiggins1981} to define cubical
$\omega$-groupoids. Here we start with more general definitions for
cubical $(\omega,p)$-categories~\cite{LucasPhD2017}.
A $k$-cell $A$ of an $\omega$-category $\Ccal$ is
\emph{$R_{i}$-invertible}, for $1 \leq i \leq k$, if there is
a $k$-cell $B$ such that
\[
A \circ_{i} B = \epsilon_{i} \partial_{i}^- A
\qquad\text{and}\qquad
B \circ_{i} A = \epsilon_{i} \partial_{i}^+ A.
\]
The $k$-cell $B$ is thus uniquely defined and denoted
$R_{i} A$, using the (partial) \emph{inversion} map $R_{i}$.
A $k$-cell $A$ has an \emph{$R_{i}$-invertible shell}, for
$1 \leq i \leq k$, if
\begin{enumerate}
\item the cells $\partial_{j}^\alpha A$ are $R_{i-1}$-invertible, for every $1\leq j<i$, 
\item the cells $\partial_{j}^\alpha A$ are $R_{i}$-invertible, for every $i<j\leq k$.
\end{enumerate}

Inverting a $k$-cell $A$ along direction $i$ swaps the faces
$\partial_{i}^-A$, $\partial_{i}^+A$ and inverts all other faces:
\begin{equation*}
\begin{tikzcd}[global scale = 2 and 2 and 1 and 1.2]
\ar[r, shorten <= -5] \ar[d, shorten <= -7] & i \\
j & 
\end{tikzcd}
\qquad
\begin{tikzcd}[global scale = 2 and 2 and 1 and 1.2]
\ar[rr] \ar[dd, shorten <= -4, "\partial_{i}^-A"'] && \ar[dd, shorten <= -4, "\partial_{i}^+A"] \\
 & A & \\
\ar[rr] && 
\end{tikzcd}
\;\;
\overset{R_{i}}{\longmapsto}
\;\;
\begin{tikzcd}[global scale = 2 and 2 and 1 and 1.2]
\ar[dd, shorten <= -4, "\partial_{i}^+A"'] && \ar[ll] \ar[dd, shorten <= -4, "\partial_{i}^-A"] \\
 & R_{i}A & \\
 && \ar[ll]
\end{tikzcd}
\end{equation*}

Using the map $R_{i}$, Lucas~\cite{LucasPhD2017} introduced an alternative inversion map
\[
T_{i} A \: := \:  \left(\epsilon_{i} \partial_{i+1}^- A \circ_{i+1} \Gamma_{i}^+ \partial_{i}^+ A\right) 
\circ_{i} \left(R_{i} \left(\Gamma_{i}^+ \partial_{i+1}^- A \circ_{i+1} a \circ_{i+1} \Gamma_{i}^- \partial_{i+1}^+ A\right) \right) 
\circ_{i} \left(\Gamma_{i}^- \partial_{i}^- A \circ_{i+1} \epsilon_{i} \partial_{i+1}^+ A\right),
\]
for all $1\leq i<k$ and every $k$-cell $A$. The $T_i$ exchange the
faces of a cell $A$ between the directions $i$ and $(i+1)$ while
applying inversion maps to all other faces:
\begin{equation*}
\begin{tikzcd}[global scale = 2.5 and 2 and 1 and 1.3]
\ar[r, shorten <= -5] \ar[d, shorten <= -7] & i+1 \\
i &
\end{tikzcd}
\qquad
\begin{tikzcd}[global scale = 2 and 2 and 1 and 1.2, longer arrows = 0.4em and 0.4em]
\phantom{\circ} \ar[rr, "\partial_{i}^- A"] \ar[dd, "\partial_{i+1}^- A"'] && \phantom{\circ} \ar[dd, "\partial_{i+1}^+ A"] \\
 & A & \\
\phantom{\circ} \ar[rr, "\partial_{i}^+ A"'] && \phantom{\circ}
\end{tikzcd}
\;\;
\overset{T_{i}}{\longmapsto}
\;\;
\begin{tikzcd}[global scale = 2 and 2 and 1 and 1.2, longer arrows = 0.4em and 0.4em]
\phantom{\circ} \ar[rr, "\partial_{i+1}^- A"] \ar[dd, "\partial_{i}^- A"'] && \phantom{\circ} \ar[dd, "\partial_{i}^+ A"] \\
 & T_{i} A & \\
\phantom{\circ} \ar[rr, "\partial_{i+1}^+ A"'] && \phantom{\circ}
\end{tikzcd},
\end{equation*}
Additional properties of inversion maps, which are needed later, are
listed in Appendix~\SSS{AA:InversionMaps}.

\subsubsection{$(\omega,p)$-categories and $\omega$-groupoids} 
An \emph{$(\omega,p)$-category} $\Ccal$ is an $\omega$-category in
which every $k$-cell with an $R_{i}$-invertible shell is
$R_{i}$-invertible for all $k>p$ and $1 \leq i \leq k$.  A
\emph{functor of $(\omega,p)$-categories} is a functor between the
underlying $\omega$-categories.  An \emph{$\omega$-groupoid} is an
$(\omega,0)$-category.

\subsubsection{Lax transformations}
\label{D:DefLaxTransCubCat}
We recall Lucas' definition of lax transformations (called \emph{lax
  $1$-transfors} by him)~\cite{Lucas2018,LucasPhD2017}. They adapt
natural transformations to cubical categories. We use them to
define contractions of $(\omega,p)$-categories in
Section~\ref{SS:CubicalNormalisationStrategies}.

 A \emph{lax transformation} $\eta:F\Rightarrow G$ between
 $(\omega,p)$-functors $F,G:\Ccal\to \Dcal$ is a family of maps that sends
 each $k$-cell $x$ in $\Ccal$ to a $(k+1)$-cell $\eta_x$ in
 $\Dcal$, for every $k\in\mathbb{N}$.  It satisfies, for all
 $1\leq i \leq k$ and $k$-cells $x,y$ in $\Ccal$,
\begin{enumerate}
\item \label{I:AxiomLaxTransCubFaces}
if $i\neq 1$ then $\partial_{1}^-\eta_x = F(x)$,
  $\partial_{1}^+\eta_x = G(x)$ and
  $\partial_{i}^\alpha\eta_x = \eta_{\partial_{i-1}^\alpha
    x}$,
\item $\eta_{x\circ_{i} y} = \eta_x\circ_{i+1}\eta_y$ if $x$ and $y$ are $i$-composable,
\item $\eta_{\epsilon_{i} z} = \epsilon_{i+1} \eta_z$ if $k<n-1$,
\item $\eta_{\Gamma_{i}^\alpha z} = \Gamma_{i+1}^\alpha \eta_z$ if $i<k<n-1$.
\end{enumerate}

Axiom~\ref{I:AxiomLaxTransCubFaces}) indicates that $\sigma_x$ is a
transformation from $F(x)$ to $G(x)$, in the sense that its source and
target faces in direction $1$ are determined by $F(x)$ and $G(x)$,
respectively. Its faces in the other directions are determined by the
value of $\sigma$ at the faces of $x$, suggesting that $\sigma$ can be
defined recursively in the dimensions. The shape of $\sigma_x$ is
\begin{equation*}
\begin{tikzcd}[global scale = 2 and 2 and 1 and 1.2]
\ar[r, shorten <= -5] \ar[d, shorten <= -7] & i \\
1 & 
\end{tikzcd}
\qquad
\begin{tikzcd}[global scale = 3.3 and 2 and 1 and 1.2]
F(\partial_{i-1}^-x) \ar[rr, "F(x)"] \ar[dd, "\sigma_{\partial_{i-1}^-x}"'] && F(\partial_{i-1}^+x) \ar[dd, "\sigma_{\partial_{i-1}^+x}"] \\
& \Downarrow\sigma_x & \\
G(\partial_{i-1}^-x) \ar[rr, "G(x)"'] && G(\partial_{i-1}^+x)
\end{tikzcd}
\end{equation*}

\section{Cubical contractions and acyclicity}
\label{S:CubContrStrat}

In this section, we introduce contractions for cubical categories, extending the
corresponding notion for globular
categories~\cite{GuiraudMalbos12advances}, and generalising the
normalisation strategies of rewriting theory to higher dimensions.
The main result in this section,
Theorem~\ref{T:ContractingImpliesAcyclic}, shows that contracting
$\omega$-groupoids are acyclic, providing a constructive
method for proving acyclicity.
 
\subsection{Contractions}
\label{SS:CubicalNormalisationStrategies}

Defining contractions for an $(\omega,p)$-category $\Ccal$
requires a notion of section, and in turn the construction of quotient
$p$-categories on $(\omega,p)$-categories.

\subsubsection{}
The face maps in the coequaliser
\begin{equation*}
\begin{tikzcd}[global scale = 6 and 5 and 1 and 1.2]
\Ccal_{p+1} \ar[r, shift left = 0.5, "\partial_{1}^-"] \ar[r, shift right = 0.5, "\partial_{1}^+"'] & \Ccal_p \ar[r, "\pi"] & \overline{\Ccal}_p
\end{tikzcd}
\end{equation*}
in the category $\catego{Set}$ compare the two faces of a $(p+1)$-cell in direction
$1$.  We could have chosen any other direction~$i$ instead to
construct $\overline{\Ccal}_p$, as the following lemma shows.

\begin{lemma}
\label{L:QuotientSectionEquivDirections}
In every cubical $(\omega,p)$-category $\Ccal$, the following coequalisers
are equal for $2\leq j\leq p+1$:
\begin{equation*}
\begin{tikzcd}[global scale = 6 and 5 and 1 and 1.2]
\Ccal_{p+1} \ar[r, shift left = 0.5, "\partial_{1}^-"] \ar[r, shift right = 0.5, "\partial_{1}^+"'] & \Ccal_p \ar[r, "\pi"] & \overline{\Ccal_p}
\end{tikzcd}
\qquad\text{ and }\qquad
\begin{tikzcd}[global scale = 6 and 5 and 1 and 1.2]
\Ccal_{p+1} \ar[r, shift left = 0.5, "\partial_{j}^-"] \ar[r, shift right = 0.5, "\partial_{j}^+"'] & \Ccal_p \ar[r, "\pi"] & \overline{\Ccal_p}.
\end{tikzcd}
\end{equation*}
\end{lemma}
\begin{proof}
  Two $p$-cells $f$, $g$ in $\Ccal$ are in the same equivalence class
  of the second coequaliser if and only if there is a $(p+1)$-cell $A$
  in $\Ccal$ such that, for all $1 \leq i \leq p+1$ such that
  $i \neq j$,
\begin{equation*}
  \partial_{j}^- A = f,
  \qquad
  \partial_{j}^+ A = g,
  \qquad
\partial_{i}^\alpha A =
\begin{cases*}
\epsilon_{j-1} \partial_{i}^\alpha f & if $i<j$, \\
\epsilon_{j} \partial_{i-1}^\alpha f & if $i>j$.
\end{cases*}
\end{equation*}

These identities assemble to the diagram
\begin{equation*}
\begin{tikzcd}[global scale = 2 and 2 and 1 and 1.2]
\ar[r, shorten <= -5] \ar[d, shorten <= -7] & i \\
j & 
\end{tikzcd}
\qquad
\begin{tikzcd}[global scale = 2 and 2 and 1 and 1.2]
x \ar[rr, "f"] \ar[dd, equal] && y \ar[dd, equal] \\
 & A & \\
x \ar[rr, "g"'] && y
\end{tikzcd}
\end{equation*}
Let $A$ be such a cell and define
$B = R_1 T_2 \dots T_{j-1} \left( \extGamma{1}{j}{+} f \circ_j A
  \circ_j \extGamma{1}{j}{-} g \right)$ where $\extGamma{1}{j}{+}$ and
$\extGamma{1}{j}{-}$ are extended connections defined as
$\extGamma{l}{m}{\alpha} = T_{m-1} \dots T_{l+1} \Gamma_l^\alpha$, for
all $l<m$. The faces of $B$ in direction $1$ are equal to $f$ and $g$;
all others are degenerate. Thus $f$ and $g$ are in the same
equivalence class for the first coequaliser. The reverse direction is similar.
\end{proof}

\subsubsection{Quotient category $ \overline{\Ccal}_p$}
\label{SSS:QuotientCat}

We equip the set $\overline{\Ccal}_p$ with face, composition,
degeneracy and connection maps. For the composition map $\circ_i$, for
$i<p$, we write $X\times_{\Ccal_i}X$ for the pullback of
$X\oto{\partial_i^-}\Ccal_i\ofrom{\partial_i^+}X$ for any set $X$. We
use the coequaliser
\begin{equation*}
\begin{tikzcd}[global scale = 10 and 5 and 1 and 1.2]
\Ccal_{p+1}\times_{\Ccal_i}\Ccal_{p+1} 
\ar[r, shift left = 0.5, "\partial_1^-\times \partial_1^-"] \ar[r, shift right = 0.5, "\partial_1^+\times \partial_1^+"'] & \Ccal_p\times_{\Ccal_i}\Ccal_p \ar[r] & \overline{\Ccal_p\times_{\Ccal_i}\Ccal_p}\simeq\overline{\Ccal}_p\times_{\Ccal_i}\overline{\Ccal}_p
\end{tikzcd}
\end{equation*}
($\simeq$ is unique because coequalisers and pullbacks commute in
$\catego{Set}$) to define
$\circ_i:\overline{\Ccal}_p\times_{\Ccal_i}\overline{\Ccal}_p\to\overline{\Ccal}_p$
as the unique map for which the diagram
\begin{equation*}
\begin{tikzcd}[global scale = 8 and 5 and 1 and 1.2]
\Ccal_{p+1}\times_{\Ccal_i}\Ccal_{p+1} \ar[r, shift left = 0.5] \ar[r, shift right = 0.5] \ar[d, "\circ_i"'] & \Ccal_p\times_{\Ccal_i}\Ccal_p \ar[r] \ar[d, "\circ_i"'] & \overline{\Ccal}_p\times_{\Ccal_i}\overline{\Ccal}_p
\ar[d, dotted, "\circ_i"] \\
\Ccal_{p+1} \ar[r, shift left = 0.5] \ar[r, shift right = 0.5] & \Ccal_p \ar[r] & \overline{\Ccal}_p
\end{tikzcd}
\end{equation*}
commutes. Face, degeneracy and connection maps are defined likewise,
using the universal property of the coequaliser $\overline{\Ccal}_p$.
This extends $\Ccal_{p-1}$ to an $(\omega,p)$-category, also denoted
$\overline{\Ccal}_p$. Its $p$-cells are equivalence classes modulo
$\Ccal_{p+1}$, and it has degenerate and connection cells in
dimensions higher than $p$.

\subsubsection{Unital sections}
\label{SSS:DefSectionCub}
The canonical projection $(\omega,p)$-functor
$\pi:\Ccal\to\overline{\Ccal}_p$ is an identity on $k$-cells for
$k<p$. It sends $p$-cells to their equivalence classes in
$\overline{\Ccal}_p$ and $k$-cells of dimension $k>p$ to degenerate
cells.  The \emph{fibre} of $\pi$ over a $p$-cell $u$ in
$\overline{\Ccal}_p$ extends to the $(\omega,0)$-category $\Ccal_u$
defined as follows:
\begin{enumerate}
\item its $0$-cells are the $p$-cells $x$ in $\Ccal$ such that $\pi(x)=u$,
\item its $k$-cells are the $(p+k)$-cells $f$ in $\Ccal$ such that
  $\partial_{p+1,1}^{\alpha_1}\partial_{p+2,1}^{\alpha_2}\dots\partial_{p+k,1}^{\alpha_k}
  f\in u$, for every $k\geq1$, 
\item its face maps $\partial'^\alpha_{k,i}$ on $\Ccal_u$ are the
  $\partial_{p+k,i}^\alpha$, for all $1\leq i\leq k$,
\item likewise for the degeneracy, connection and composition maps.
\end{enumerate}

A \emph{section} of the projection $\pi:\Ccal\to\overline{\Ccal}_p$ is a family 
\begin{equation*}
\iota=(\iota_u:\onebf\to\Ccal_u)_{u\in\overline{\Ccal}_p}
\end{equation*}
of $(\omega,0)$-functors, where $\onebf$ is the terminal category in
$\Ccal$.  We only consider \emph{unital} sections, which satisfy
$\iota_{\pi(t)}=t$ for every thin $p$-cell $t$ in $\Ccal$ and for all
$p\geq1$, but usually omit this adjective.

The section $\iota$ sends each $p$-cell $u$ in $\overline{\Ccal}_p$ to
a functor $\iota_u$ with the representative $p$-cell of $u$ in $\Ccal$
in its image, while leaving all thin cells unchanged.  We write
$\iota_u$ for this representative of $u$ as well.  Moreover, for every
$k$-cell $f$ of $\Ccal$ with $p\leq k$ we write $\widehat{f}$ for the
image of $\iota_{\pi(f)}$ in $\Ccal_{\pi(f)}$ by abuse of
notation,. Example diagrams for sections are given
in~\SSS{SSS:IllustrContr}.

\subsubsection{Contractions}
\label{SSS:ContractionsCub}
Let $\iota$ be a section of the projection $\pi:\Ccal\to\overline{\Ccal}_p$. 
A \emph{$\iota$-contraction} of $\Ccal$  is a family~$\sigma$ of lax transformations
\begin{equation*}
\left(
\begin{tikzcd}[global scale = 3 and 2 and 1 and 1.2]
\Ccal_u \ar[rr, bend left, "\id", ""'{name=U}] \ar[dr, bend right, "\zeta"'] && \Ccal_u \\
& |[alias=D]| \onebf \ar[ur, bend right, "\iota_u"'] &
\ar[from=U, to=D, Rightarrow, "\sigma_u"]
\end{tikzcd}
\right)_{u\in\overline{\Ccal}_p}\,,
\end{equation*}
where $\zeta$ is the unique $(\omega,0)$-functor into $\onebf$, such that
\begin{eqn}{equation}
\label{E:ConditionContractions}
\sigma_{\iota_u}=\epsilon_{1}\iota_u
\qquad\text{and}\qquad
\sigma_{\sigma_f}=\Gamma_{1}^-\sigma_f, 
\end{eqn}
for each $u$ in $\overline{\Ccal}_p$ and $f$ in $\Ccal_k$ with
$p\leq k$, and where $\sigma_g$ stands for $(\sigma_{\pi(g)})_g$ for
each cell $g$ in $\Ccal_{\ell}$ for $p \leq \ell$.
  
Expanding this definition, a $\iota$-contraction $\sigma$ is a family
of maps $(\Ccal_k \to \Ccal_{k+1})_{k\ge p}$ such that for each
$k$-cells $f,g$ in $\Ccal$ and every $i$ with $p+1\leq i\leq k$, the
conditions {{\bf i)}-{\bf iv)} from \SSS{D:DefLaxTransCubCat} hold:
\begin{enumerate}
\item The \emph{boundary} $\partial(\sigma_f)$ is the $(k-1)$-square
  $f^\partial$ defined by
\begin{eqn}{equation*}
\partial_1^-f^\partial=f,
\qquad
\partial_1^+f^\partial=\epsilon_k\dots\epsilon_{p+1}\widehat{x},
\qquad
\partial_i^\alpha f^\partial=\sigma_{\partial_{i-1}^\alpha f},
\end{eqn}
which yields the diagram
\begin{equation*}
\begin{tikzcd}[global scale = 2 and 2 and 1 and 1.2]
\ar[r, shorten <= -5] \ar[d, shorten <= -7] & i \\
1 & 
\end{tikzcd}
\qquad
f^\partial = \begin{tikzcd}[global scale = 2 and 2 and 1 and 1.2]
  x \ar[rr, "f"] \ar[dd, "\sigma_x"'] && y \ar[dd, "\sigma_y"] \\
  && \\
  \widehat{x} \ar[rr, equal] && \widehat{y}
\end{tikzcd}
\end{equation*}
\item If $f$ and $g$ are $\circ_i$-composable, then
\begin{equation*}
\begin{tikzcd}[global scale = 1.3 and 1 and 1 and 1.2]
\ar[rr, shorten <= -5] \ar[dd, shorten <= -7] && i+1 \\\\
1 & 
\end{tikzcd}
\qquad
\sigma_{f \circ_i g} = \sigma_f\circ_{i+1}\sigma_g =
\begin{tikzcd}[global scale = 2 and 2 and 1 and 1.2]
x \ar[rr, "f"] \ar[dd, "\sigma_x"'] && y \ar[rr, "g"] \ar[dd, "\sigma_y"{pos=0.3}] && z \ar[dd, "\sigma_z"] \\
 & \sigma_f && \sigma_g & \\
\widehat{x} \ar[rr, equal] && \widehat{y} \ar[rr, equal] && \widehat{z}
\end{tikzcd}
\end{equation*}
\item
\begin{equation*}
\begin{tikzcd}[global scale = 1.3 and 1 and 1 and 1.2]
 & i+2 & \\
\ar[ur, shorten <= -3] \ar[rr, shorten <= -5] \ar[dd, shorten <= -7] && i+1 \\\\
1 && 
\end{tikzcd}
\qquad
\sigma_{\epsilon_i f} = \epsilon_{i+1} \sigma_f =
\begin{tikzcd}[global scale = 1.2 and 1.2 and 1 and 1.2]
 && y \ar[rrrr, equal] \ar[dddd] &&&& y \ar[dddd] \\
 &&& \epsilon_i f &&& \\
x \ar[uurr, "f"] \ar[rrrr, equal, crossing over] \ar[dddd] &&&& x
\ar[uurr] && \\
 & \sigma_f &&&&& \\
 && \widehat{y} \ar[rrrr, equal] &&&& \widehat{y} \\\\
\widehat{x} \ar[uurr, equal] \ar[rrrr, equal] &&&& \widehat{y}
\ar[uurr, equal] \ar[from=uuuu, crossing over]&& 
\end{tikzcd}
\end{equation*}
\item If $i<k$, then
\begin{equation*}
\begin{tikzcd}[global scale = 1.3 and 1 and 1 and 1.2]
 & i+2 & \\
\ar[ur, shorten <= -3] \ar[rr, shorten <= -5] \ar[dd, shorten <= -7] && i+1 \\\\
1 && 
\end{tikzcd}
\qquad
\sigma_{\Gamma_i^\alpha f} = \Gamma_{i+1}^\alpha \sigma_f =
\begin{tikzcd}[global scale = 1.2 and 1.2 and 1 and 1.2]
 && y \ar[rrrr, equal] \ar[dddd] &&&& y \ar[dddd] \\
 &&& \Gamma_i^\alpha f &&& \\
x \ar[uurr, "f"] \ar[rrrr, crossing over] \ar[dddd] &&&& y \ar[uurr, equal] && \\
 & \sigma_f &&&&& \\
 && \widehat{y} \ar[rrrr, equal] &&&& \widehat{y} \\\\
\widehat{x} \ar[uurr, equal] \ar[rrrr, equal] &&&& \widehat{y}
\ar[uurr, equal]\ar[from=uuuu, crossing over] && 
\end{tikzcd}
\end{equation*}
\end{enumerate}

In addition, the second condition in~\eqref{E:ConditionContractions}
expands as follows: $\sigma_{\sigma_f}$ is the thin cell
\begin{equation*}
\begin{tikzcd}[global scale = 1 and 1 and 1 and 1.2]
 & i & \\
\ar[ur, shorten <= -3] \ar[rr, shorten <= -5] \ar[dd, shorten <= -7] && 2 \\\\
1 && 
\end{tikzcd}
\qquad
\sigma_{\sigma_f} = \Gamma_1^- \sigma_f =
\begin{tikzcd}[global scale = 1.2 and 1.2 and 1 and 1.2]
 && y \ar[rrrr] \ar[dddd] &&&& \widehat{y} \ar[dddd, equal] \\
 &&& \sigma_f &&& \\
x \ar[uurr, "f"] \ar[rrrr ,crossing over] \ar[dddd] &&&& \widehat{x} \ar[uurr, equal] && \\
 & \sigma_f &&&&& \\
 && \widehat{y} \ar[rrrr, equal] &&&& \widehat{y} \\\\
\widehat{x} \ar[uurr, equal] \ar[rrrr, equal] &&&& \widehat{x}
\ar[uurr, equal]\ar[from=uuuu, equal, crossing over] && 
\end{tikzcd}
\end{equation*}

The first condition  in~\eqref{E:ConditionContractions} is equivalent to
$\sigma_{\widehat{x}}=\epsilon_{1}\widehat{x}$ for each $p$-cell
$x$ in $\Ccal$:
\begin{equation*}
\begin{tikzcd}[global scale = 4 and 2 and 1 and 1.2]
\widehat{x} \ar[r, equal, "\sigma_{\widehat{x}}"] & \widehat{x}
\end{tikzcd}
\end{equation*}

Examples of contractions in low dimensions are given
in~\SSS{SSS:IllustrContr}. Contractions, understood as families of lax
transformations, can be computed recursively across all dimensions,
starting from a chosen section. They are also compatible with
inverses, as stated in the following lemma.

\begin{lemma}
For every $k$-cell $f$ with $p\leq k$, and for all $1\leq i\leq k$ and $1\leq j<k$, 
\[
\sigma_{R_i f} = R_{i+1} \sigma_f
\qquad\text{ and }\qquad
\sigma_{T_j f} = T_{j+1} \sigma_f.
\]
\end{lemma}

\begin{proof}
  For $\sigma_{R_if}=R_{i+1}\sigma_f$, we check that
  $\sigma_{R_if}\circ_{i+1}\sigma_f$ and
  $\sigma_f\circ_{i+1}\sigma_{R_if}$ are thin cells.  The claim then
  holds because thin cells with the same boundaries are
  equal~\cite{LucasPhD2017}.  The proof of
  $\sigma_{T_j f} = T_{j+1} \sigma_f$ is similar.
\end{proof}

\subsubsection{}
\label{SSS:IllustrContr}
We present example diagrams for $f^\partial$ and $\sigma_f$ for a cell
$f$ of low dimension in an $\omega$-groupoid $\Ccal$.

\begin{enumerate}
\item If $x\in \Ccal_0$, then $x^\partial$ is the $0$-square
  $(x,\widehat{x})$ and $\sigma_x : x \fl \widehat{x}$ the $1$-cell
  filling it.
\item If $f\in \Ccal_1$, then $f^\partial$ is an $1$-square and
  $\sigma_f$ a $2$-cell filling it:
\begin{equation*}
\begin{tikzcd}[global scale = 2 and 2 and 1 and 1.2]
\ar[r, shorten <= -5] \ar[d, shorten <= -7] & 2 \\
1 & 
\end{tikzcd}
\qquad
\begin{array}{cc}
f^\partial =
\begin{tikzcd}[global scale = 2 and 2 and 1 and 1.2]
x \ar[rr, "f"] \ar[dd, "\sigma_x"'] && y \ar[dd, "\sigma_y"] \\\\
\widehat{x} \ar[rr, equal] && \widehat{y}
\end{tikzcd}
 & \qquad
\begin{tikzcd}[global scale = 2 and 2 and 1 and 1.2]
x \ar[rr, "f"] \ar[dd, "\sigma_x"'] && y \ar[dd, "\sigma_y"] \\
 & \sigma_f & \\
\widehat{x} \ar[rr, equal] && \widehat{y}
\end{tikzcd}
\end{array}
\end{equation*}
\item If $A\in \Ccal_2$, then $A^\partial$ is a
  $2$-square and $\sigma_A$ a $3$-cell
  filling it:
\begin{equation*}
\begin{tikzcd}[global scale = 1 and 1 and 1 and 1.2]
 & 3 & \\
\ar[ur, shorten <= -2] \ar[rr, shorten <= -5] \ar[dd, shorten <= -7] && 2 \\\\
1 && 
\end{tikzcd}
\qquad
\begin{array}{cc}
A^\partial =
\begin{tikzcd}[global scale = 1.2 and 1.2 and 1 and 1.2]
 && y_3 \ar[rrrr] \ar[dddd] &&&& y \ar[dddd] \\
 &&& A &&& \\
x \ar[uurr] \ar[rrrr, crossing over] \ar[dddd] &&&& y_2 \ar[uurr] && \\
 & \sigma_{\partial_1^- A} &&&& \sigma_{\partial_1^+ A} & \\
 && \widehat{y_3} \ar[rrrr, equal] &&&& \widehat{y} \\\\
\widehat{x} \ar[uurr, equal] \ar[rrrr, equal] &&&& \widehat{y_2}
\ar[uurr, equal]\ar[from=uuuu, crossing over] && 
\end{tikzcd}
 & \qquad\qquad 
\begin{tikzcd}[global scale = 1.2 and 1.2 and 1 and 1.2]
 && y_3 \ar[rrrr] \ar[dddd] &&&& y \ar[dddd] \\\\
x \ar[uurr] \ar[rrrr, crossing over] \ar[dddd] &&&& y_2 \ar[uurr] && \\
 &&& \sigma_A &&& \\
 && \widehat{y_3} \ar[rrrr, equal] &&&& \widehat{y} \\\\
\widehat{x} \ar[uurr, equal] \ar[rrrr, equal] &&&& \widehat{y_2}
\ar[uurr, equal]\ar[from=uuuu, crossing over] && 
\end{tikzcd}
\end{array}
\end{equation*}
\end{enumerate}

\subsubsection{}
An $(\omega,p)$-category is \emph{contracting} if it admits a
contraction.
This property does not depend on particular choices of sections.  For each
$\widehat{(-)}$-contraction $\sigma$, we can define a
$\tilde{(-)}$-contraction $\tau$ such that, for every $k$-cell $f$,
with $p\leq k < n$, the $(k+1)$-cell $\tau_f$ is the composition
\begin{equation*}
\tau_f = \sigma_f\circ_1R_1\sigma_{\epsilon_k\dots\epsilon_{p+1}\tilde{x}}
= \sigma_f\circ_1\epsilon_{k+1}\dots\epsilon_{p+2}R_1\sigma_{\tilde{x}},
\end{equation*}
where $x=\partial_{p+1}^-\dots\partial_k^-f$.  For $p=0$ and
$x\in\Ccal_0$, for instance, $\tau_x$ is the $\circ_1$-composition
\begin{equation*}
\begin{tikzcd}[global scale = 4.6 and 1.5 and 1 and 1.2]
x \ar[r, "\sigma_x"] & \widehat{x} \ar[r, "R_1\sigma_{\tilde{x}}"] & \tilde{x}
\end{tikzcd}
\end{equation*}
and for $f\in\Ccal_1$, $\tau_f$ is the $\circ_1$-composition
\begin{equation*}
\begin{tikzcd}[global scale = 3.7 and 2 and 1 and 1.2]
x \ar[rr, "f"] \ar[dd] && y \ar[dd] \\
 & \sigma_f & \\
\widehat{x} \ar[rr, equal] \ar[dd] && \widehat{y} \ar[dd] \\
 & R_1\sigma_{\epsilon_1\tilde{x}}=\epsilon_2R_1\sigma_{\tilde{x}} & \\
\tilde{x} \ar[rr, equal] && \tilde{y}
\end{tikzcd}
\end{equation*}

\subsection{Acyclic $\omega$-groupoids}
\label{SS:ContractionAcyclicity}

We now show that acyclicity of $\omega$-groupoids can be obtained by
constructing contractions. Our proof unfolds cubes into cubes with
degenerate faces in each direction $i \geq 2$ using folding and
unfolding maps~\cite[Def. 3.1]{AlAglBrownSteiner2002}.

\subsubsection{Acyclicity}
\label{SSS:SquaresFillersAcyclicity}

Defining acyclicity for a cubical $(\omega,p)$-category
$\Ccal$ requires three further notions:

\begin{enumerate}
\item A $k$-\emph{square} of $\Ccal$, for $k\geq0$, is a family
$(f_i^\alpha)_{1\leq i\leq k+1,\alpha}$ of $k$-cells in $\Ccal$ such
that
\begin{eqn}{equation}
\label{E:SquareEquations}
\partial_{i}^\alpha f_j^\beta=\partial_{j-1}^\beta f_i^\alpha,
\end{eqn}
for all $1\leq i<j\leq k+1$. We write $\Squa{k}(\Ccal)$ for the
set of $k$-squares of $\Ccal$.

\item The \emph{boundary} $\partial A$ of a $k$-cell $A$ in $\Ccal$,
  for $k \geq 1$, is the $(k-1)$-square
  $(\partial_i^\alpha A)_{1\leq i\leq k,\alpha}$.

\item A \emph{filler} of a $k$-square $S$ is a $(k+1)$-cell $A$ such that
$\partial A =S$. 
\end{enumerate}
An $(\omega,p)$-category $\Ccal$ is \emph{acyclic} if, for $k\geq p$,
every $k$-square of $\Ccal$ has a filler.

The following diagrams show  a $2$-cell $A$ and its boundary $1$-square:
\begin{equation*}
\begin{tikzcd}[global scale = 2 and 2 and 1 and 1.2]
\ar[r, shorten <= -5] \ar[d, shorten <= -7] & 2 \\
1
\end{tikzcd}
\qquad
\begin{tikzcd}[global scale = 2 and 2 and 1 and 1.2]
 \ar[rr, "\partial_{1}^- A"] \ar[dd, "\partial_{2}^- A"'] && \ar[dd, "\partial_{2}^+ A"] \\
 & \Downarrow A & \\
 \ar[rr,"\partial_{1}^+ A"'] && 
\end{tikzcd}
\qquad\quad
\partial A = 
\begin{tikzcd}[global scale = 2 and 2 and 1 and 1.2]
 \ar[rr, "\partial_{1}^- A"] \ar[dd, "\partial_{2}^- A"'] && \ar[dd, "\partial_{2}^+ A"] \\\\
 \ar[rr,"\partial_{1}^+ A"'] && 
\end{tikzcd}
\end{equation*}

\subsubsection{Folding and unfolding}
\label{SSS:FoldingUnfolding}
Let $\Ccal$ be an $\omega$-category. The \emph{folding maps}
$\psi_i,\Psi_j,\Phi_k:\Ccal_m\to\Ccal_m$ are defined, for
$1\leq i\leq m-1$, $1\leq j\leq m$ and $0\leq k\leq m$ as
\begin{align*}
\psi_i(x) & =  \Gamma_i^+ \partial_{i+1}^- x \circ_{i+1} x \circ_{i+1} \Gamma_i^- \partial_{i+1}^+ x
= 
\vcenter{\hbox{
\begin{tikzpicture}[global scale = 1 and 1]
\node [] () at (1.5,0.5) {$x$};
\draw [-] (0,0) -- (3,0);
\draw [-] (0,1) -- (3,1);
\draw [-] (0,0) -- (0,1);
\draw [-] (1,0) -- (1,1);
\draw [-] (2,0) -- (2,1);
\draw [-] (3,0) -- (3,1);
\draw [-] (0.5,0.2) -- (0.5,0.5) -- (0.8,0.5);
\draw [-] (2.2,0.5) -- (2.5,0.5) -- (2.5,0.8);
\end{tikzpicture}
}}\;,\\
\Psi_j &=  \begin{cases*}
	\id & if $j=1$, \\
	\psi_{j-1} \Psi_{j-1} & otherwise
	\end{cases*}	
= 
\psi_{j-1} \psi_{j-2} \dots \psi_1,\\
\Phi_k & =  \begin{cases*}
	\id & if $k=0$, \\
	\Phi_{k-1} \Psi_k & otherwise
	\end{cases*}
=  
\Psi_1 \Psi_2 \dots \Psi_k
= 
\psi_1 (\psi_2 \psi_1) \dots (\psi_{k-1} \dots \psi_1).
\end{align*}
They extend to maps from $(m-1)$-squares to 
$(m-1)$-squares~\cite[Prop. 8.5]{AlAglBrownSteiner2002}.

Consider the sets
\[
  \SquaFill{m-1}{\varphi}  =  \{(S,A)\in\Squa{m-1}(\Ccal)\times\Ccal_m
  \mid \partial A=\varphi(S)\}
\]
of squares with corresponding fillers, for $\varphi\in\{\psi_i,\Psi_j,\Phi_k\}$.
The \emph{unfolding maps}
$\overline{\psi}_i:\SquaFill{m-1}{\psi_i}\to\Ccal_m$,
$\overline{\Psi}_j:\SquaFill{m-1}{\Psi_j}\to\Ccal_m$ and
$\overline{\Phi}_k:\SquaFill{m-1}{\Phi_k}\to\Ccal_m$ are defined as
\begin{align*}
\overline{\psi}_i(S,A)
& = 
(\epsilon_i S_i^- \circ_{i+1} \Gamma_i^+ S_{i+1}^+) 
 \circ_i A 
 \circ_i (\Gamma_i^- S_{i+1}^- \circ_{i+1} \epsilon_i S_i^+)
= 
\vcenter{\hbox{
\begin{tikzpicture}[global scale = 1 and 1]
\node [] () at (1.5,1.5) {$A$};
\draw [-] (0,0) -- (3,0);
\draw [-] (0,1) -- (3,1);
\draw [-] (0,2) -- (3,2);
\draw [-] (0,3) -- (3,3);
\draw [-] (0,0) -- (0,3);
\draw [-] (1,0) -- (1,1);
\draw [-] (1,2) -- (1,3);
\draw [-] (2,0) -- (2,1);
\draw [-] (2,2) -- (2,3);
\draw [-] (3,0) -- (3,3);
\draw [-] (0.2,0.5) -- (0.5,0.5) -- (0.5,0.8);
\draw [-] (1.5,0.2) -- (1.5,0.8);
\draw [-] (1.5,2.2) -- (1.5,2.8);
\draw [-] (2.5,2.2) -- (2.5,2.5) -- (2.8,2.5);
\end{tikzpicture}
}}\\
\overline{\Psi}_j(S,A) & = 
	\begin{cases*}
	A & if $j=1$, \\
	\overline{\Psi}_{j-1}(S,\overline{\psi}_{j-1}(\Psi_{j-1}(S),A)) & otherwise,
	\end{cases*}
\\
\overline{\Phi}_k(S,A) & = 
	\begin{cases*}
	A & if $k=0$, \\
	\overline \Psi_k(S,\overline\Phi_{k-1}(\Psi_k(S),A)) & otherwise.
	\end{cases*}
\end{align*}

\begin{lemma}
\label{L:UnFoldingProperties}
Every folding or unfolding map
$\overline\varphi\in\{\overline\psi_i,\overline\Psi_j,\overline\Phi_k\}$
satisfies $\partial \overline{\varphi}(S,A) = S$, for every
$(m-1)$-square $S$ and $m$-cell $A$ such that
$\partial A = \varphi(S)$.
\end{lemma}
\begin{proof}
  The proof of $\overline{\psi}_i$ is straightforward.  Those
  for $\overline{\Psi}_j$ and $\overline{\Phi}_k$ follow by induction.
\end{proof}

We are now prepared for the main result of this section.
\begin{theorem}
\label{T:ContractingImpliesAcyclic}
Every contracting $\omega$-groupoid is acyclic.
\end{theorem}
\begin{proof}
  Suppose $\Ccal$ is an $\omega$-groupoid with a section
  $\widehat{(-)}$ of the projection $\pi:\Ccal\to\overline{\Ccal}_0$ and a contraction~$\sigma$.  For $m\geq 2$, let $S$ be
  an $(m-1)$-square. We set $T=\Phi_m(S)$, $g^\alpha = T_1^\alpha$ and
  $A = \sigma_{g^-} \circ_1 R_1\sigma_{g^+}$. Then
  \begin{equation*}
    T_k^\alpha = \epsilon_1 \partial_1^- T_k^\alpha = \epsilon_1
  \partial_1^+ T_k^\alpha = \epsilon_1 \partial_k^\alpha g^- =
  \epsilon_1 \partial_k^\alpha g^+
\end{equation*}
for every $1<k\leq m$ by
  \cite[Prop. 3.6]{AlAglBrownSteiner2002}.  It follows that
  $\partial_k^\alpha g^- = \partial_k^\alpha g^+$, for every
  $1<k\leq m$. Hence $A$ is a filler of $T$, because, for $1<k\leq m$,
\begin{equation*} \partial_k^\alpha A =
    \partial_k^\alpha \sigma_{g^-} \circ_1 R_1 \partial_k^\alpha
    \sigma_{g^+} = \sigma_{\partial_{k-1}^\alpha g^-} \circ_1 R_1
    \sigma_{\partial_{k-1}^\alpha g^+} = \sigma_{\partial_{k-1}^\alpha
      g^-} \circ_1 R_1 \sigma_{\partial_{k-1}^\alpha g^-} = \epsilon_1
    \partial_1^- \sigma_{\partial_{k-1}^\alpha g^-} = T_k^\alpha,
\end{equation*}
and the case $k=1$ is obvious. Finally, set
$B=\overline{\Phi}_m(S,A)$. By the above calculation and
Lemma~\ref{L:UnFoldingProperties}, $\partial B=S$, that is, $B$
is a filler of $S$ and acyclicity of $\Ccal$ follows.
\end{proof}

\subsubsection{The case $n=2$}
Theorem~\ref{T:ContractingImpliesAcyclic} remains valid for
$n$-groupoids with $n\geq 2$. First, the definitions of sections and
contraction in~\SSS{SSS:DefSectionCub}
and~\SSS{SSS:ContractionsCub} extend to $n$-groupoids, forgetting
all cells of dimension greater than $n$. The proof replays that for
$\omega$-groupoids, except that only $(m-1)$-squares with
$2\leq m\leq n$ require consideration. As an example, we show that
every contracting $2$-groupoid $\Ccal$ is acyclic.  Suppose $\Ccal$
has a $\widehat{(-)}$-contraction $\sigma$.  We start with a
$1$-square
\begin{equation*}
S=
\begin{tikzcd}[global scale = 4 and 4 and 1 and 1.2]
a \ar[r, "S_1^-"] \ar[d, "S_2^-"'] & b \ar[d, "S_2^+"] \\
c \ar[r, "S_1^+"'] & d
\end{tikzcd}.
\end{equation*}
The folding maps yield the $1$-square
\begin{equation*}
T=\Phi_2(S)=\Psi_2(S)=\psi_1(S)=
\begin{tikzcd}[global scale = 2 and 2 and 1 and 1.2]
a \ar[equal, rr] \ar[equal, dd] && a \ar[rr, "S_1^-"] \ar[dd] && b \ar[rr, "S_2^+"] \ar[dd] && d \ar[equal, dd] \\
& \Gamma && S && \rotatebox[origin=c]{180}{$\Gamma$} & \\
a \ar[rr, "S_2^-"'] && c \ar[rr, "S_1^+"'] && d \ar[equal, rr] && d
\end{tikzcd}
=
\begin{tikzcd}[global scale = 4 and 4 and 1 and 1.2]
a \ar[r, "S_1^-"] \ar[d, equal] & b \ar[r, "S_2^+"] & d \ar[d, equal] \\
a \ar[r, "S_2^-"'] & c \ar[r, "S_1^+"'] & d
\end{tikzcd}
.
\end{equation*}\
The contraction $\sigma$ fills the $1$-square $T$ with
\begin{equation*}
A=
\begin{tikzcd}[global scale = 4 and 2 and 1 and 1.2]
a \ar[rr, "T_1^- = S_1^- \circ_1 S_2^+"] \ar[dd, "\sigma_x"'] && d \ar[dd, "\sigma_{x'}"] \\
 & \sigma_{T_1^-} & \\
\widehat{a} \ar[rr, equal] && \widehat{d} \\
 & R_1\sigma_{T_1^+} & \\
a \ar[uu, "\sigma_x"] \ar[rr, "T_1^+ = S_2^- \circ_1 S_1^+"'] && d \ar[uu, "\sigma_{x'}"']
\end{tikzcd}
=
\begin{tikzcd}[global scale = 4 and 2 and 1 and 1.2]
a \ar[rr, "T_1^-"] \ar[dd, equal] && d \ar[dd, equal] \\
 & \sigma_{T_1^-} \circ_1 R_1\sigma_{T_1^+} & \\
a \ar[rr, "T_1^+"'] && d
\end{tikzcd}
.
\end{equation*}
The unfolding maps then allow us to construct the following filler of
$S$, showing that $\Ccal$ is acyclic:
\begin{equation*}
B=\overline{\Phi}_2(S,A)=\overline{\Psi}_2(S,A)=\overline{\psi}_1(S,A)=
\begin{tikzcd}[global scale = 2 and 2 and 1 and 1.2]
a \ar[equal, rr] \ar[equal, dd] && a \ar[rr, "S_1^-"] \ar[equal, dd] && b \ar[equal, rr] \ar[equal, dd] && b \ar[dd, "S_2^+"] \\
&&& \vert && \Gamma & \\
a \ar[equal, rr] \ar[equal, dd] && a \ar[rr] && b \ar[rr] && d \ar[equal, dd] \\
&&& A &&& \\
a \ar[rr] \ar[dd, "S_2^-"'] && c \ar[rr] \ar[equal, dd] && d \ar[equal, rr] \ar[equal, dd] && d \ar[equal, dd] \\
& \rotatebox[origin=c]{180}{$\Gamma$} && \vert &&& \\
c \ar[equal, rr] && c \ar[rr, "S_1^+"'] && d \ar[equal, rr] && d
\end{tikzcd}
\end{equation*}

  Lucas has established a variant of
  Theorem~\ref{T:ContractingImpliesAcyclic} for cubical monoidal
  $(2,0)$-poly\-graphs~\cite{Lucas20}. Instead of using folding and
  unfolding maps, he rotates cells with the same shape as
  contractions with the inversion maps $R_{i}$ and $T_{i}$, and then
  glues them using connection maps. Folding and unfolding maps seem to
  make the proof for cubical $\omega$-groupoids easier. These
  maps rotate all the faces of cubes in direction $1$, so that the
  proof does not become more difficult with increasing dimension.

\section{Cubical coherent confluence}
\label{S:CCConfluence}

We now use the cubical machinery introduced in the previous section to
establish confluence properties of abstract rewriting systems (ARS) in
cubical $(p+2)$- or $(p+3)$-categories, for any $p\in\Nbb$.  Although
cubes have $(p+2)$ dimensions, we restricted rewriting relations in
two or three fixed directions. Apart from coherent versions of
Newman's lemma and the Church-Rosser theorem in two directions, we
also prove Newman's lemma also in three directions, for which
additional structure was present or a specific cube law had to be
imposed
previously~\cite{LevyPhD78,Barendregt84,EndrullisKlop2019,Klop2022}.
The coherence in these results expresses the way to tile confluence or local confluence diagrams by pasting a given set of higher-dimensional witnesses.
Finally, as a special case of Theorem~\ref{T:ContractingImpliesAcyclic}, we derive a cubical version
of Squier’s theorem, using normalisation strategies as special kinds
of sections and contractions.  We assume familiarity with the basics
of classical
rewriting~\cite{DershowitzJouannaud90,Klop92,BaaderNipkow98,Terese03}.

\subsection{Confluence fillers}
\label{SS:ConfluenceFillers}

\subsubsection{Abstract rewriting in cubical categories}
\label{SSS:RewritingSystemInCubCat}

Let $\Ccal$ be a $(p+2)$-category for some $p\in\Nbb$.
We fix an integer $i$ such that $1\leq i\leq p-1$, representing a choice of direction.
A \emph{$p$-ARS in $\Ccal$} is a subset~$\Xcal_{\Ccal}$ of $\Ccal_{p+1}$, whose elements are non-degenerate in direction $i$.
We write $\cstabs{\Xcal_{\Ccal}}$ (resp. $\cstabsr{\Xcal_{\Ccal}}$) 
for the smallest subsets of $\Ccal_{p+1}$ that contain $\Xcal_{\Ccal}$ 
and are stable under $\circ_i$-compositions (resp. $\circ_i$-compositions and inversions). 
The elements of $\cstabs{\Xcal_{\Ccal}}$ are sequences $(f_1,\ldots,f_k)$ of $\circ_i$-composable $(p+1)$-cells in $\Ccal$, called \emph{rewriting paths of length $k$}, which we identify with their composite $f_1\circ_i \ldots \circ_i f_k$ in $\Ccal$. 
The elements of~$\cstabsr{\Xcal_{\Ccal}}$ are sequences of $\circ_i$-composable $(p+1)$-cells in $\Ccal$ and their $R_i$-inverses, called \emph{rewriting zigzags}.

The $p$-ARS $\Xcal_{\Ccal}$ is \emph{Noetherian} (in direction $i$) if
it admits no rewriting path of infinite length.  This property is
needed for proofs by induction on rewriting paths.

A \emph{branching (in direction $i$)} of $\Xcal_{\Ccal}$ is a
pair $(f_1,f_2)$ of $(p+1)$-cells in $\cstabs{\Xcal_{\Ccal}}$ such that
$\partial_{i}^-f_1=\partial_{i}^-f_2$. It is \emph{local} if $f_1,f_2\in\Xcal_{\Ccal}$.
We denote by $\Brc{}{\Xcal_{\Ccal}}$ (resp. $\LBrc{}{\Xcal_{\Ccal}}$) the set of branchings (resp. local branchings) of $\Xcal_{\Ccal}$.
The $p$-ARS $\Xcal_{\Ccal}$ is \emph{(locally) confluent (in
  direction $i$)} if for every (local) branching $(f_1,f_2)$ of $\Xcal_{\Ccal}$,
there are $g_1,g_2\in\cstabs{\Xcal}$ such that
\begin{equation*}
\partial_{i}^+f_1 = \partial_{i}^-g_1,
\qquad
\partial_{i}^+f_2 = \partial_{i}^-g_2,
\qquad
\partial_{i}^+g_1 = \partial_{i}^+g_2.
\end{equation*}
These identities determine the \emph{confluence diagram}
\begin{eqn}{equation}
\label{E:ConflDiag}
\begin{tikzcd}[global scale = 2.5 and 2 and 1 and 1.2]
\ar[r, shorten <= -5] \ar[d, shorten <= -7] & i+1 \\
i & 
\end{tikzcd}
\qquad
\begin{tikzcd}[global scale = 2 and 2 and 1 and 1.2]
x \ar[rr, "f_2"] \ar[dd, "f_1"'] && y_2 \ar[dd, "g_2"] \\
&& \\
y_1 \ar[rr, "g_1"'] && z
\end{tikzcd}
\end{eqn}
The $p$-ARS $\Xcal_{\Ccal}$ is \emph{convergent} if it is confluent and Noetherian.

\subsubsection{Confluence fillers}
A \emph{(local) confluence filler (in direction $i$)} of a (local) branching
$(f_1,f_2)$ of $\Xcal_{\Ccal}$ is a $(p+2)$-cell $\filler{2}{f_1,f_2}$ in $\Ccal$ such that
\begin{equation*}
\partial_{i}^-\filler{2}{f_1,f_2} = f_2,
\qquad
\partial_{i+1}^-\filler{2}{f_1,f_2} = f_1,
\qquad
\partial_{i}^+\filler{2}{f_1,f_2},\partial_{i+1}^+\filler{2}{f_1,f_2} \in \cstabs{\Xcal_{\Ccal}}.
\end{equation*}
This determines a (local) confluence diagram similar to~\eqref{E:ConflDiag}:
\begin{equation*}
\begin{tikzcd}[global scale = 2.25 and 2 and 1 and 1.2]
\ar[r, shorten <= -5] \ar[d, shorten <= -7] & i+1 \\
i & 
\end{tikzcd}
\qquad
  \begin{tikzcd}[global scale = 2.5 and 2.1 and 1 and 1.2]
  x \ar[rr, "f_2"] \ar[dd, "f_1", swap] && y_2 \ar[dd,
  "\partial_{i+1}^+\filler{2}{f_1,f_2}"] \\
  & \filler{2}{f_1,f_2} & \\
  y_1 \ar[rr, "\partial_{i}^+\filler{2}{f_1,f_2}", swap] && z
  \end{tikzcd}
\end{equation*}
We write $\LCf{}{\Xcal_{\Ccal}}$ (resp. $\Cf{}{\Xcal_{\Ccal}}$) for
the subset of $\Cr_{p+2}$ of cells with the shape of a local confluence
filler (resp. confluence filler), that is, $(p+2)$-cells $A$ such that
\[
\partial_i^+(A),\partial_{i+1}^+(A)\in \Xcal_{\Ccal}^{\circ_i},
\qquad
\partial_i^-(A),\partial_{i+1}^-(A)\in \Xcal_{\Ccal},
\quad
\text{(resp.}\;\;
\partial_i^-(A),\partial_{i+1}^-(A)\in \Xcal_{\Ccal}^{\circ_i}).
\]
Therefore, $A_2$ defines a map 
$A_2 : \LBrc{}{\Xcal_{\Ccal}} \fl \LCf{}{\Xcal_{\Ccal}}$ 
(resp. $A_2 : \Brc{}{\Xcal_{\Ccal}} \fl \Cf{}{\Xcal_{\Ccal}}$).

We can now state and prove a coherent cubical version of Newman’s
lemma.

\begin{proposition}
\label{P:CubicalCoherentNewmanVersionDimP}
For a Noetherian $p$-ARS $\Xcal_{\Ccal}$, each map $A_2$ extends from
$\LBrc{}{\Xcal_{\Ccal}} \fl \LCf{}{\Xcal_{\Ccal}}$ to~$\Brc{}{\Xcal_{\Ccal}} \fl \Cf{}{\Xcal_{\Ccal}}$.
\end{proposition}
\begin{proof}
  We extend the map $A_2$ by Noetherian induction in direction $i$ on
  the source of branchings. We order $p$-cells by the relation
  $\leq$ generated by $\Xcal_{\Ccal}$, defined by $x\leq y$ if there
  is a rewriting path $f$ such that $\partial_i^-f=y$ and
  $\partial_i^+f=x$.

  The base case is trivial. For the induction step, le $(f_1,f_2)$ be
  a branching. If $f_1$ is a degeneracy in direction $i$, the result
  is trivial, as the map $A_2$ is extended by the formula
  $\filler{2}{f_1,f_2}=\epsilon_if_2$. The case where $f_2$ is a
  degeneracy in direction $i$ is similar. In the other cases, we
  decompose $f_1=g_1\circ_{i}h_1$ and $f_2=g_2\circ_{i}h_2$, with
  $g_1,g_2\in\Xcal_{\Ccal}$ and $h_1,h_2\in\cstabs{\Xcal_{\Ccal}}$ and
  extend $A_2$ recursively as
\[
\filler{2}{f_1,f_2} = \left(\filler{2}{g_1,g_2} \circ_{i+1} \filler{2}{\partial_i^+\filler{2}{g_1,g_2},h_2}\right) \circ_i \filler{2}{h_1,\partial_i^+\left(\filler{2}{g_1,g_2}\circ_{i+1}\filler{2}{\partial_i^+\filler{2}{g_1,g_2},h_2}\right)}.
\]
This pasting of cubes ressembles the classical diagrammatic
proof of Newman's lemma:
\begin{eqn}{equation}
\label{E:CubicalSchemeNewman2}
\begin{tikzcd}[global scale = 2.5 and 2 and 1 and 1]
\ar[r, shorten <= -6] \ar[d, shorten <= -8] & i+1 \\
i &
\end{tikzcd}
\hspace{6em}
\begin{tikzcd}[global scale = 5 and 2 and 1 and 1.3, longer arrows = 2 and 2]
x \ar[rr, "g_2"] \ar[dd, "g_1"'] && x' \ar[rr, "h_2"] \ar[dd] && x'' \ar[dd] \\
& \filler{2}{g_1,g_2} && \filler{2}{\partial_i^+\filler{2}{g_1,g_2},h_2} & \\
y \ar[rr] \ar[dd, "h_1"'] && y' \ar[rr] && y'' \ar[dd] \\
&& \filler{2}{h_1,\partial_i^+\left(\filler{2}{g_1,g_2}\circ_{i+1}\filler{2}{\partial_i^+\filler{2}{g_1,g_2},h_2}\right)} && \\
z \ar[rrrr] &&&& z'
\end{tikzcd}
\end{eqn}
\end{proof}

\subsubsection{Church-Rosser fillers} 
A \emph{Church-Rosser filler (in direction $i$)} of a cell $f$ in $\cstabsr{\Xcal_{\Ccal}}$ is a $(p+2)$-cell $\fillerCR{f}$ in $\Ccal$ such that 
\[
\partial_{i}^-\fillerCR{f} = f,
\qquad
\partial_{i}^+\fillerCR{f} = \epsilon_{i}\partial_{i}^+\partial_{i+1}^+\fillerCR{f},
\qquad
\partial_{i+1}^-\fillerCR{f},\partial_{i+1}^+\fillerCR{f} \in \cstabs{\Xcal},
\]
which determines the \emph{Church-Rosser diagram}
\begin{equation*}
\begin{tikzcd}[global scale = 2.5 and 2 and 1 and 1.2]
\ar[r, shorten <= -5] \ar[d, shorten <= -7] & i+1 \\
i & 
\end{tikzcd}
\qquad
\begin{tikzcd}[global scale = 2 and 2 and 1 and 1.2]
x \ar[rr, "f"] \ar[dd, "\partial_{i+1}^-\fillerCR{f}", swap] && y \ar[dd, "\partial_{i+1}^+\fillerCR{f}"] \\
& \fillerCR{f} & \\
z \ar[rr, equal] && z
\end{tikzcd}
\end{equation*}
Once again this correspondence defines a map
$B : \cstabsr{\Xcal_{\Ccal}} \fl \CR{\Xcal_{\Ccal}}$, were
$\CR{\Xcal_{\Ccal}}$ denotes the subset of $\Ccal_{p+2}$ of cells of
the shape of a Church-Rosser fillers.

With these definitions, we prove a coherent cubical version of the
Church-Rosser theorem.

\begin{proposition}
\label{P:CubicalCoherentChurchRosserVersionDimP}
For a $p$-ARS $\Xcal_{\Ccal}$ in a $(p+2,p+1)$-category $\Ccal$, each
map $A_2 : \Brc{}{\Xcal_{\Ccal}} \fl \Cf{}{\Xcal_{\Ccal}}$ induces a
map $B : \cstabsr{\Xcal_{\Ccal}} \fl \CR{\Xcal_{\Ccal}}$.
\end{proposition}
\begin{proof}
  Every cell $f$ in $\cstabsr{\Xcal_{\Ccal}}$ is an zigzag
  $f_1\circ_{i}\dots\circ_{i}f_k$ of minimal length $k$ of
  non-$\circ_{i}$-identity cells in $\cstabs{\Xcal_{\Ccal}}$ and of
  $R_{i}$-inverses of such cells.  We define the map $B$ on cells of
  $\cstabsr{\Xcal_{\Ccal}}$ by induction on their length $k$.  The
  base case $k = 1$ is trivial. For the induction step, for $k\geq2$
  and $f_1\in\cstabs{\Xcal_{\Ccal}}$, we extend~$B$ recursively:
\begin{equation*}
\fillerCR{f} = \left(\Gamma_{i}^-f_1\circ_{i+1}\epsilon_{i}(f_2\circ_{i}\dots\circ_{i}f_k)\right)\circ_{i}\fillerCR{f_2\circ_{i}\dots\circ_{i}f_k},
\end{equation*}
which corresponds to the diagram
\begin{eqn}{equation}
\begin{tikzcd}[global scale = 2.5 and 2 and 1 and 1]
\ar[r, shorten <= -6] \ar[d, shorten <= -8] & i+1 \\
i & 
\end{tikzcd}
\hspace{4em}
\begin{tikzcd}[global scale = 2.5 and 2 and 1 and 1.3]
x \ar[rr, "f_1"] \ar[dd, "f_1"'] && x' \ar[rrrr, leftrightarrow, "f_2\circ_{i}\dots\circ_{i}f_k"] \ar[dd, equal] &&&& x'' \ar[dd, equal] \\
& \Gamma_{i}^- &&& \epsilon_{i} && \\
x' \ar[rr, equal] \ar[dd] && x' \ar[rrrr, leftrightarrow, "f_2\circ_{i}\dots\circ_{i}f_k"] &&&& x'' \ar[dd] \\
&&& \fillerCR{f_2\circ_{i}\dots\circ_{i}f_k} &&& \\
y \ar[rrrrrr, equal] &&&&&& y
\end{tikzcd}
\end{eqn}

Otherwise, for $k\geq2$ and $R_{i}f_1\in\cstabs{\Xcal_{\Ccal}}$, we
extend $B$ recursively using the map $A_2$:
\begin{gather*}
\fillerCR{f} = \left(\epsilon_{i}f_1\circ_{i+1}\Gamma_{i}^+f_1\circ_{i+1}\fillerCR{f_2\circ_{i}\dots\circ_{i}f_k}\right)
\circ_{i}\left(R_{i+1}\Gamma_{i}^-f_1\circ_{i+1}\filler{2}{f_1,g}\right)
\circ_{i}\Gamma_{i}^-\partial_{i}^+\filler{2}{f_1,g},
\end{gather*}
where $g$ denotes the $(p+1)$-cell
$\partial_{i+1}^-\fillerCR{f_2\circ_{i}\dots\circ_{i}f_k}$. This
corresponds to the diagram
\begin{eqn}{equation}
\begin{tikzcd}[global scale = 2.5 and 2 and 1 and 1]
\ar[r, shorten <= -6] \ar[d, shorten <= -8] & i+1 \\
i & 
\end{tikzcd}
\hspace{4em}
\begin{tikzcd}[global scale = 2.5 and 2 and 1 and 1.3]
x \ar[dd, equal] && x' \ar[ll, "f_1"'] \ar[rr, equal] \ar[dd, equal] && x' \ar[dd, "g"'] \ar[rrrr, leftrightarrow, "f_2\circ_{i}\dots\circ_{i}f_k"] &&&& x'' \ar[dd] \\
& \epsilon_{i} && \Gamma_{i}^+ &&& \fillerCR{f_2\circ_{i}\dots\circ_{i}f_k} && \\
x \ar[dd, equal] && x' \ar[ll] \ar[rr] \ar[dd] && y \ar[rrrr, equal] &&&& y \ar[dd] \\
& R_{i+1}\Gamma_{i}^- &&&& \filler{2}{f_1,g} &&& \\
x \ar[rr, equal] \ar[dd] && x \ar[rrrrrr] &&&&&& z \ar[dd, equal] \\
&&&& \Gamma_{i}^- &&&& \\
z \ar[rrrrrrrr, equal] &&&&&&&& z
\end{tikzcd}
\end{eqn}
\end{proof}

The diagrams in the proof of the coherent Church-Rosser theorem reduce
to the familiar triangular shapes in the classical diagrammatic
Church-Rosser proof once degeneracies are collapsed and the
corresponding $p$-cells are identified.

\subsection{$3$-Confluence and the cube law}
\label{SS:ThreeConflucenceCubeLaw}

\subsubsection{$3$-confluence fillers}
A \emph{$3$-branching (in direction $i$)} of a $p$-ARS $\Xcal_{\Ccal}$ in a $(p+3)$-category $\Ccal$ is a triple $(f_1,f_2,f_3)$ of $(p+1)$-cells in $\cstabs{\Xcal_{\Ccal}}$ such that $\partial_{i}^-f_1=\partial_{i}^-f_2=\partial_{i}^- f_3$. 
It is \emph{local} if $f_1,f_2,f_3\in\Xcal_{\Ccal}$.
We denote by $\Brc{3}{\Xcal_{\Ccal}}$ (resp. $\LBrc{3}{\Xcal_{\Ccal}}$) the set of $3$-branchings (resp. local $3$-branchings) of $\Xcal_{\Ccal}$.

A \emph{(local) $3$-confluence filler} with respect to a map $A_2 : \Brc{}{\Xcal_{\Ccal}} \fl \Cf{}{\Xcal_{\Ccal}}$ of a (local) $3$-branching $(f_1,f_2,f_3)$ is a $(p+3)$-cell
$\filler{3}{f_1,f_2,f_3}$ in $\Ccal$ with faces 
\begin{align*}
\partial_{i}^-\filler{3}{f_1,f_2,f_3} & =  \filler{2}{f_2,f_3},
&\partial_{i}^+\filler{3}{f_1,f_2,f_3} & =  \filler{2}{\partial_{i}^+\filler{2}{f_1,f_2},\partial_{i}^+\filler{2}{f_1,f_3}}, \\
\partial_{i+1}^-\filler{3}{f_1,f_2,f_3} & =  \filler{2}{f_1,f_3},
&\partial_{i+1}^+\filler{3}{f_1,f_2,f_3} & =  \filler{2}{\partial_{i+1}^+\filler{2}{f_1,f_2},\partial_{i}^+\filler{2}{f_2,f_3}}, \\
\partial_{i+2}^-\filler{3}{f_1,f_2,f_3} & =  \filler{2}{f_1,f_2},
&\partial_{i+2}^+\filler{3}{f_1,f_2,f_3} & =  \filler{2}{\partial_{i+1}^+\filler{2}{f_1,f_3},\partial_{i+1}^+\filler{2}{f_2,f_3}}.
\end{align*}
We write $\Cf{3}{\Xcal_{\Ccal},A_2}$
(resp. $\LCf{3}{\Xcal_{\Ccal},A_2}$) for the set of confluence fillers
(resp. local confluence fillers) with respect to the map $A_2$.

These definitions allow us to prove a coherent cubical
Newman's lemma in three directions, and thus in three dimensions.

\begin{proposition}
\label{P:CubicalCoherentNewmanDim3VersionDimP}
Let $\Xcal_{\Ccal}$ be a Noetherian $p$-ARS in a $(p+3)$-category
$\Ccal$ with a map $A_2 : \LBrc{}{\Xcal_{\Ccal}} \fl
\LCf{}{\Xcal_{\Ccal}}$. Then each map $A_3 $ extends from $\LBrc{3}{\Xcal_{\Ccal}} \fl
\LCf{3}{\Xcal_{\Ccal},A_2}$ to $\Brc{3}{\Xcal_{\Ccal}} \fl \Cf{3}{\Xcal_{\Ccal},A_2}$.
\end{proposition}

\begin{proof}
  By Proposition~\ref{P:CubicalCoherentNewmanVersionDimP}, the map
  $A_2$ extends from local to arbitrary branchings and
    confluences.  We extend the map $A_3$ by induction in direction
  $i$ on the source of the $3$-branchings.

The base case is trivial. 
Let $(f_1,f_2,f_3)$ be a $3$-branching with source $x$ and suppose that the map $A_3$ extends to $3$-branchings with source a $p$-cell reduced from $x$.
For each $1\leq i\leq 3$, we write $f_i = f_i'\circ_i f_i''$, where $f_i'$ belongs to $\Xcal_{\Ccal}$. 
By assumption, the local $3$-branching $(f_1',f_2',f_3')$ is filled by
the $3$-confluence filler $B=\filler{3}{f_1',f_2',f_3'}$. 
Then, using the induction hypothesis,
\begin{itemize}
\item the $3$-branching $\left(\partial_{i+1}^-\partial_{i+2}^+B,\partial_{i}^-\partial_{i+2}^+B,f_3''\right)$ is filled by the $3$-confluence filler
\begin{equation*}
C\:=\:\filler{3}{\partial_{i+1}^-\partial_{i+2}^+B,\partial_{i}^-\partial_{i+2}^+B,f_3''},
\end{equation*}
\item the $3$-branching $\left(\partial_{i+1}^-\partial_{i+1}^+B,f_2'',\partial_i^-\partial_{i+1}^+(B\circ_{i+2}C)\right)$ is filled by the $3$-confluence filler
\begin{equation*}
D\:=\:\filler{3}{\partial_{i+1}^-\partial_{i+1}^+B,f_2'',\partial_i^-\partial_{i+1}^+(B\circ_{i+2}C)},
\end{equation*}
\item the $3$-branching $\left(f_1'',\partial_{i+1}^-\partial_i^+\left((B\circ_{i+2}C)\circ_{i+1}D\right),\partial_i^-\partial_i^+\left((B\circ_{i+2}C)\circ_{i+1}D\right)\right)$ is filled by the $3$-confluence filler
\begin{equation*}
E\:=\:\filler{3}{f_1'',\partial_{i+1}^-\partial_i^+\left((B\circ_{i+2}C)\circ_{i+1}D\right),\partial_i^-\partial_i^+\left((B\circ_{i+2}C)\circ_{i+1}D\right)}.
\end{equation*}
\end{itemize}
We then extend the map $A_3$ inductively,  setting
\begin{equation*}
\filler{3}{f_1,f_2,f_3} \:=\: ((B\circ_{i+2}C)\circ_{i+1}D)\circ_{i}E.
\end{equation*}
This construction corresponds to the diagram
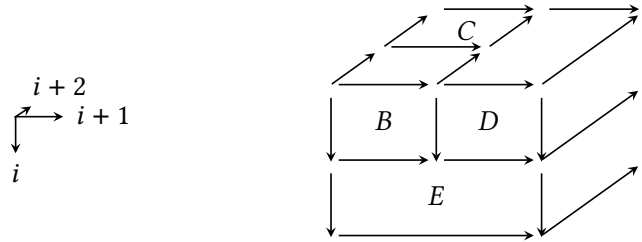
\begin{equation*}
\label{E:CubicalNewmanDimThreeVersionDimP}
\begin{tikzcd}[global scale = 1.5 and 1 and 1 and 1]
& i+2 & \\
\ar[ur, shorten <= -3] \ar[rr, shorten <= -5] \ar[dd, shorten <= -7] && i+1 \\\\
i &&
\end{tikzcd}
\hspace{6em}
\begin{tikzcd}[global scale = 1.8 and 1.3 and 1 and 1, longer arrows = 2 and 2]
&& |[alias=UL]| \ar[rr] && \ar[rr] && \ar[dd] \\
& \ar[ur] \ar[rr] && |[alias=DR]| \ar[ur] &&& \\
\ar[ur] \ar[rr] \ar[dd] && \ar[ur] \ar[rr] \ar[dd] && \ar[uurr] \ar[dd] && \ar[dd] \\
& B && D &&& \\
\ar[rr] \ar[dd] && \ar[rr] && \ar[uurr] \ar[dd] && \\
&& E &&&& \\
\ar[rrrr] &&&& \ar[uurr] &&
\ar[from=UL, to=DR, phantom, "C"]
\end{tikzcd}
\end{equation*}
\end{proof}

\subsubsection{The cube law}
\label{SSS:ResiduationConflCubeLaw}
Our functional approach to confluence fillers admits an interpretation
in terms of residual paths and of the cube law.  Indeed, the map $A_2$
allows defining a \emph{residuation} operation 
\begin{equation*}
  \res{f_1}{f_2} \:\coloneq\: \partial_{i+1}^+\filler{2}{f_1,f_2},
  \end{equation*}
  for every branching $(f_1,f_2)$. This operation is well-known from
  the $\lambda$-calculus~\cite{LevyPhD78,Barendregt84}.  It gives rise
  to the confluence diagram
 \begin{equation*}
\begin{tikzcd}[global scale = 2 and 2 and 1 and 1.3]
x \ar[rr, "f_2"] \ar[dd, "f_1"'] && y_2 \ar[dd, "\res{f_1}{f_2}"] \\\\
y_1 \ar[rr, "\res{f_2}{f_1}"'] && z
\end{tikzcd}
\end{equation*}
To work with residuals, it helps
  memorising $\res{f_1}{f_2}$ as the translation of $f_1$ along
  $f_2$ in the square spanned by $f_1$ and $f_2$, and $\res{f_2}{f_1}$ as the translation of $f_2$ along $f_1$.

Lévy has shown that residuation satisfies the
cube law in $\lambda$-calculus \cite[Lemma~2.2.1]{LevyPhD78},
see also~\cite[Lemma~12.2.6]{Barendregt84}
and~\cite[Def. 4.49]{Dehornoy2015}, which is often presented as a
  single cube law up to permutation of indices.  
For a $3$-branching $(f_1,f_2,f_3)$, the \emph{cube law} state that
\begin{eqn}{equation*}
(\res{f_i}{f_j})(\res{f_k}{f_j}) \: = \: (\res{f_i}{f_k})(\res{f_j}{f_k}),
\end{eqn}
for all pairwise distinct $i,j,k$ in $\{1,2,3\}$.  Geometrically,
this law assembles rewriting paths around the following cube spanned by the $3$-branching:
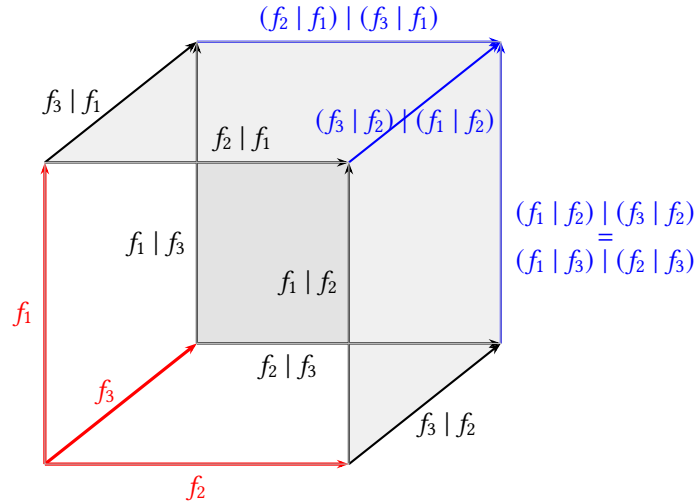
\begin{eqn}{equation*}
\raisebox{-4cm}{
\begin{tikzpicture}[
    scale=1,
    arrows={-Stealth[length=4.5pt, width=3.2pt]},
    every node/.style={font=\small},
    gen/.style={very thick, red},
    resi/.style={thick, blue},
    resin/.style={thick, black},
]
  %--- Coordonnées du cube (carré avant ABCD, carré arrière A'B'C'D')
  \coordinate (A)  at (0,0);      % avant-bas-gauche
  \coordinate (B)  at (4,0);      % avant-bas-droite
  \coordinate (C)  at (4,4);      % avant-haut-droite
  \coordinate (D)  at (0,4);      % avant-haut-gauche

  % 💡 Profondeur augmentée ici :
  \coordinate (shift) at (2,1.6);

  \coordinate (Aprime) at ($(A)+(shift)$); % arrière-bas-gauche
  \coordinate (Bprime) at ($(B)+(shift)$); % arrière-bas-droite
  \coordinate (Cprime) at ($(C)+(shift)$); % arrière-haut-droite
  \coordinate (Dprime) at ($(D)+(shift)$); % arrière-haut-gauche

  %--- Arrière-plan : faces arrière grisées, face avant transparente
  \begin{pgfonlayer}{background}
    \fill[gray!22] (Aprime) -- (Bprime) -- (Cprime) -- (Dprime) -- cycle;
    \fill[gray!12] (B) -- (Bprime) -- (Cprime) -- (C) -- cycle;
    \fill[gray!12] (D) -- (Dprime) -- (Cprime) -- (C) -- cycle;
  \end{pgfonlayer}

  %--- Générateurs (rouge)
  \draw[gen] (A) -- node[pos=0.5, left] {$f_1$} (D);
  \draw[gen] (A) -- node[pos=0.5, below] {$f_2$} (B);
  \draw[gen] (A) -- node[pos=0.4, above] {$f_3$} (Aprime);

  %--- Résiduels (bleu)
  \draw[resin] (D) -- node[pos=0.65, above] {$\res{f_2}{f_1}$} (C);
  \draw[resin] (B) -- node[pos=0.6, left] {$\res{f_1}{f_2}$} (C);
  \draw[resin] (D) -- node[pos=0.5, left] {$\res{f_3}{f_1}$\;} (Dprime);
  \draw[resin] (B) -- node[pos=.33, right] {\;$\res{f_3}{f_2}$} (Bprime);
  \draw[resin] (Aprime) -- node[pos=0.3, below] {$\res{f_2}{f_3}$} (Bprime);
  \draw[resin] (Aprime) -- node[pos=.33, left] {$\res{f_1}{f_3}$} (Dprime);

  \draw[resi] (Dprime) -- node[pos=0.5, above] {$\res{(\res{f_2}{f_1})}{(\res{f_3}{f_1})}$} (Cprime);

  \draw[resi] (Bprime) -- node[pos=0.35, right] {$\begin{array}{c}
     \res{(\res{f_1}{f_2})}{(\res{f_3}{f_2})}\\[-0.3em]
     =\\[-0.3em]
     \res{(\res{f_1}{f_3})}{(\res{f_2}{f_3})}
    \end{array}$} (Cprime);

  \draw[resi] (C) -- node[pos=.53,below] {\!\!\!\!\!\!\!\!\!\!\!$\res{(\res{f_3}{f_2})}{(\res{f_1}{f_2})}$} (Cprime);

  %--- Contours
  \draw[gray!50] (A) -- (D) -- (C) -- (B) -- cycle;
  \draw[gray!50] (Aprime) -- (Dprime) -- (Cprime) -- (Bprime) -- cycle;

\end{tikzpicture}
}
\end{eqn}
In this cube, the residual path
$f_1\mathop{|} f_2$ is obtained by translating $f_1$ along $f_2$ in
the front square and the residual path $f_3\mathop{|} f_2$ by
translating $f_3$ along $f_2$ in the bottom square, so that
$(f_1\mathop{|}f_2)\mathop{|}(f_3\mathop{|}f_2)$ is the residual path
of these two residual paths on the back face of the cube. Similar
translations show that
$(f_1\mathop{|}f_3)\mathop{|}(f_2\mathop{|}f_3)$ represents the same
arrow. The other instantiations of $f_1$, $f_2$, $f_3$ in the cube law
produce the remaining blue arrows and thus
assemble all arrows around the cube.
  
The cube law follows from the cubical
relations~\eqref{E:AxiomPreCubClass} applied to the cube
$A_3(f_1,f_2,f_3)$, for instance, 
\begin{equation*}
(f_1\mathop{|}f_2)\mathop{|}(f_3\mathop{|}f_2) = \partial_{i+1}^+\partial_{i+1}^+\filler{3}{f_1,f_2,f_3} =
\partial_{i+1}^+\partial_{i+2}^+\filler{3}{f_1,f_2,f_3} = (f_1\mathop{|}f_3)\mathop{|}(f_2\mathop{|}f_3).
\end{equation*}
They are thus a natural and immediate consequence of the way faces are attached to cells of cubical sets, hence of the geometry of cubes
that emerges somewhat accidentally from the laws of
$\lambda$-calculus. The cube law has appeared more
recently as a postulate in $3$-confluence proofs in classical abstract
rewriting~\cite{EndrullisKlop2019,Klop2022}.

\subsection{Normalising confluence}
\label{SS:NormalisingConfluence}

Next we bring the sections and contractions from
Section~\ref{S:CCConfluence} into play and prove normalising variants
of Newman's lemma and the Church-Rosser theorem. We also state and
prove a cubical version of Squier's theorem, which requires
normalisation.

\subsubsection{Normal forms and contractions}
\label{SSS:NormalFormContraction}
Let  $\Xcal_{\Ccal}$ be a convergent $p$-ARS in a $(p+2)$-category $\Ccal$.
A cell $x\in\Ccal_{p}$ is a \emph{normal form (in direction $i$)} if there are no cells $f\in\Xcal_{\Ccal}$ for which $\partial_{i}^-f = x$.
By convergence, any rewriting path that starts from any $x\in\Ccal_{p}$
terminates in a unique normal form $\widehat{x}$.
This determines a section $\widehat{(-)}$ of the projection $\pi:\Ccal\to\overline{\Ccal}_p$, as defined in~\SSS{SSS:DefSectionCub}.
For every $x\in\Ccal_{p}$ we
choose a $(p+1)$-cell $\sigma_x\in\cstabs{\Xcal_{\Ccal}}$ such that
$\sigma:\Ccal_{p}\to\Ccal_{p+1}$ is a contraction in the
$(p+1,p)$-category generated by $\Ccal_{p+1}$, as defined in~\SSS{SSS:ContractionsCub}.

\subsubsection{Normalising fillers}
\label{SSS:NormalisingConfluence}
A \emph{normalising (local) confluence filler (in direction $i$)} of a (local) branching $(f_1,f_2)$ of $\Xcal_{\Ccal}$ is a $(p+2)$-cell $\filler{2}{f_1,f_2}$ in $\Ccal$ such that
\begin{align*}
\partial_{i}^-\filler{2}{f_1,f_2} = f_2,\qquad
\partial_{i+1}^-\filler{2}{f_1,f_2} = f_1,\qquad
\partial_{i}^+\filler{2}{f_1,f_2} = \sigma_{\partial_{i}^+f_1},\qquad
\partial_{i+1}^+\filler{2}{f_1,f_2} = \sigma_{\partial_{i}^+f_2}.
\end{align*}
These identities assemble to  a (local) confluence diagram
\begin{eqn}{equation}
\label{E:shapeA2}
\begin{tikzcd}[global scale = 2.5 and 2 and 1 and 1]
\ar[r, shorten <= -5] \ar[d, shorten <= -7] & i+1 \\
i & 
\end{tikzcd}
\hspace{4em}
\begin{tikzcd}[global scale = 2 and 2 and 1 and 1.3]
x \ar[rr, "f_2"] \ar[dd, "f_1"'] && y_2 \ar[dd, "\sigma_{y_2}"] \\
& \filler{2}{f_1,f_2} & \\
y_1 \ar[rr, "\sigma_{y_1}"'] && \widehat{x}
\end{tikzcd}
\end{eqn}
A \emph{normalising Church-Rosser filler (in direction $i$)} of a cell
$f$ in $\cstabsr{\Xcal_{\Ccal}}$ is a $(p+2)$-cell $\fillerCR{f}$ in
$\Ccal$ of shape
\begin{equation*}
\begin{tikzcd}[global scale = 2.5 and 2 and 1 and 1.2]
\ar[r, shorten <= -5] \ar[d, shorten <= -7] & i+1 \\
i & 
\end{tikzcd}
\qquad
\begin{tikzcd}[global scale = 2 and 2 and 1 and 1.2]
x \ar[rr, "f"] \ar[dd, "\sigma_x"'] && y \ar[dd, "\sigma_y"] \\
& \fillerCR{f} & \\
\widehat{x} \ar[rr, equal] && \widehat{x}
\end{tikzcd}
\end{equation*}

We write $\NCf{3}{\Xcal_{\Ccal},A_2}$
(resp. $\NLCf{3}{\Xcal_{\Ccal},A_2}$) for the set of normalising
confluence fillers (resp. normalising local confluence fillers) with
respect to~$A_2$ and $\NCR{\Xcal_{\Ccal}}$ the set of normalising
Church-Rosser fillers of $\Xcal_{\Ccal}$.

These notions allow us to prove normalising variants of Newman's lemma
and the Church-Rosser theorem with the same diagrams as before, but
with normal forms and degeneracies in suitable places.

\begin{lemma}
\label{L:CoherentNormalNewmanClassical}
For a Noetherian $p$-ARS $\Xcal_{\Ccal}$, each map $A_2$
  extends from $\LBrc{}{\Xcal_{\Ccal}} \fl \NLCf{}{\Xcal_{\Ccal}}$ to $\Brc{}{\Xcal_{\Ccal}} \fl \NCf{}{\Xcal_{\Ccal}}$.
\end{lemma}
\begin{proof}
  The proof is similar to that of
  Proposition~\ref{P:CubicalCoherentNewmanVersionDimP}, but confluence
  fillers are now normalising. In
  Diagram~\eqref{E:CubicalSchemeNewman2} we thus replace $y'$, $y''$
  and $z'$ by $\widehat{x}$ and arrows between the $\widehat{x}$ by
  degeneracies.
\end{proof}

\begin{lemma}
\label{L:CoherentNormalChurchRosserClassical}
For a $p$-ARS $\Xcal_{\Ccal}$ in a $(p+2,p+1)$-category $\Ccal$,
each map $A_2 : \Brc{}{\Xcal_{\Ccal}} \fl \NCf{}{\Xcal_{\Ccal}}$ induces a 
map $B : \cstabsr{\Xcal_{\Ccal}} \fl \NCR{\Xcal_{\Ccal}}$.
\end{lemma}
\begin{proof}
  By the obvious replacements in the diagrams in the proof of
  Proposition~\ref{P:CubicalCoherentChurchRosserVersionDimP}.
\end{proof}

Finally, we state a cubical version of Squier’s theorem~\cite{SquierOttoKobayashi94} for $1$-groupoids.
Its formulation motivates the extension of the notion of cubical normalization strategies to higher dimensions, which is the subject of the next section.

\begin{proposition}
  \label{P:AcyclicityNormalisationLowDimCat}
For a convergent ARS  $\Xcal_{\Ccal}$ each map $A_2 :
   \LBrc{}{\Xcal_{\Ccal}} \fl \NLCf{}{\Xcal_{\Ccal}}$ extends to a
   witness $2$-cell for a proof of acyclicity of the groupoid ${\Xcal_{\Ccal}}^{\top_1}$.
\end{proposition}
\begin{proof}
  Lemmas~\ref{L:CoherentNormalNewmanClassical}
  and~\ref{L:CoherentNormalChurchRosserClassical} allow us to
  construct a map $B$ from zigzags to normalising Church-Rosser
  fillers. Every square $S$ is then obtained by the following
  composition of cubes:
\begin{eqn}{equation*}
\begin{tikzcd}[global scale = 2 and 2 and 1 and 1]
\ar[r, shorten <= -5] \ar[d, shorten <= -7] & 2 \\
1 & 
\end{tikzcd}
\hspace{4em}
\begin{tikzcd}[global scale = 3.5 and 2 and 1 and 1.3]
x \ar[rr, equal] \ar[dd, equal] && \ar[rr, "\partial_{2,1}^-S"] \ar[dd] && \ar[rr, equal] \ar[dd] && y_2 \ar[dd, equal] \\
& \Gamma_{2,1}^+ && \fillerCR{\partial_{2,1}^-S} && R_{2,2}\Gamma_{2,1}^+ & \\
\ar[rr] \ar[dd, "\partial_{2,2}^-S"'] && \widehat{x} \ar[rr, equal] \ar[dd, equal] && \widehat{x} \ar[dd, equal] && \ar[ll] \ar[dd, "\partial_{2,2}^+S"] \\
& T_{2,1}\fillerCR{\partial_{2,2}^-S} &&&& R_{2,2}T_{2,1}\fillerCR{\partial_{2,2}^+S} & \\
\ar[rr] \ar[dd, equal] && \widehat{x} \ar[rr, equal] && \widehat{x} && \ar[ll] \ar[dd, equal] \\
& R_{2,1}\Gamma_{2,1}^+ && R_{2,1}\fillerCR{\partial_{2,1}^+S} && R_{2,1}R_{2,2}\Gamma_{2,1}^+ & \\
y_1 \ar[rr, equal] && \ar[uu] \ar[rr, "\partial_{2,1}^+S"'] && \ar[uu] \ar[rr, equal] && y
\end{tikzcd}
\end{eqn}
\end{proof}

\begin{remark}
  Proposition~\label{P:AcyclicityNormalisationLowDimCat} is a
  low-dimensional version of
  Theorem~\ref{T:ContractingImpliesAcyclic}, proved without using
  folding maps. The same method has been used by Lucas~\cite{Lucas20},
  rotating and gluing confluence fillers to fill a square. Yet
  extending to higher dimensions as in
  Theorem~\ref{T:ContractingImpliesAcyclic} seems combinatorially
  difficult, as it requires rotating and gluing all confluence fillers
  of the faces of a $k$-square.

  For a converse of Squier's theorem for cubical $\omega$-groupoids
  freely generated by $(\omega,0)$-polygraphs see
  Theorem~\ref{T:AcyclicityNormalisation} below.
\end{remark}

\subsubsection{The cube law revisited}
\label{SSS:ExSkolemNormConfl}
Contractions allow defining $(f\mathop{|}g)=\sigma_{\partial_i^+(g)}$ for any
branching $(f,g)$. For $3$-branchings $(f_1,f_2,f_3)$, we can
then derive the cube law,
\begin{align*}
(f_1\mathop{|}f_2)\mathop{|}(f_3\mathop{|}f_2)
= \sigma_{\partial_i^+(\sigma_{\partial_i^+(f_2)})} 
= \epsilon_i\widehat{\partial_i^+(f_2)} 
= \epsilon_i\widehat{\partial_i^+(f_3)} 
= \sigma_{\partial_i^+(\sigma_{\partial_i^+(f_3)})} 
= (f_1\mathop{|}f_3)\mathop{|}(f_2\mathop{|}f_3),
\end{align*}
without using $3$-confluence fillers explicitly. Contractions also
allow constructing $3$-confluence fillers more easily, and extending
the techniques in this section to higher dimensions. In
Section~\ref{S:PolygraphicResolutionARS} we formalise
higher-dimensional versions of normalising confluence diagrams,
generated from the confluence of $n$-branchings, in cubical
$n$-polygraphs and for $n\geq2$. We use them further to construct
$\omega$-groupoids on convergent ARS.

\section{Cubical groupoids in abstract rewriting}
\label{S:PolygraphicResolutionARS}

In this section, we present extensions and applications of
Theorem~\ref{T:ContractingImpliesAcyclic} to cubical polygraphs, after
briefly recalling their structure in
Subsection~\ref{SS:CubicalPolygraphs}.
Theorem~\ref{T:AcyclicityNormalisation} shows that a free
$\omega$-groupoid on a polygraph is acyclic if and only if it is
contracting. In Subsection~\ref{SS:PolResolFromConfl}, we construct
free $\omega$-groupoids extending convergent ARS, defining for
each $k \geq 2$ a map $A_k$ from local $k$-branchings to local
$k$-confluence fillers and thereby accounting for the confluence of
higher-dimensional branchings. Finally, in
Theorem~\ref{T:TruncatedAcyclicExtensionARS}, we show that a suitable
choice of $2$-cells for $A_2$ refines this construction so that all
$k$-cells are thin for $k\geq 2$. This shows that abstract rewriting
with normalisation strategies does not require the generation of
coherence cells in dimensions higher than~$2$.  Together, these two
results provide cubical analogues of Squier’s theorem for ARSs.

\subsection{Cubical polygraphs, contractions and acyclicity of cubical
  groupoids}
\label{SS:ContractionAcyclicityPol}
\label{SS:CubicalPolygraphs}

First we recall the notion of \emph{cubical polygraph}.  The existence
of this structure was originally established by Lucas in the context
of Gray categories~\cite{LucasPhD2017} . Yet the explicit
construction of the free category generated by a cubical polygraph was
not made explicit therein. For completeness, we provide such a
construction while deferring a detailed
development to Appendix~\ref{A:CubPolFreeCat}, including a
  proof of existence of the free (cubical)
$(n-1)$-groupoid $\tck{X_{n-1}}$.

\subsubsection{Cubical polygraphs}
\label{D:DefCubPol}
A \emph{cubical $(1,0)$-polygraph} (a \emph{$1$-polygraph})
$(X_0,X_1)$ consists of a set $X_0$ of \emph{$0$-generators} and a set
$X_1$ of \emph{$1$-generators} or \emph{rewriting steps}, equipped with
\emph{source} and \emph{target} maps
$\partial_{1,1}^{\alpha} :X_1\to X_0$. 
%Two $1$-generators $f,g\in X_1$ are \emph{composable} if $\partial_{1,1}^+(f)=\partial_{1,1}^-(g)$.  
It freely generates a $1$-category $\cat{X}$, as well as a $1$-groupoid $\tck{X}$.
A \emph{cubical cellular extension} of a cubical $(n-1,0)$-category $\Ccal$ is
a set $X_n$ of \emph{$n$-generators} and face maps
$\partial_{n,i}^\alpha:X_n\to\Ccal_{n-1}$ for $1\leq i\leq n$ which
satisfy the cubical relations~\eqref{E:AxiomPreCubClass}.

A \emph{cubical $(n,0)$-polygraph} $X=(X_0,\dots,X_n)$ is formed by a
cubical $(n-1,0)$-polygraph $(X_0,\dots,X_{n-1})$ and a cubical
cellular extension $X_n$ of the free (cubical) $(n-1)$-groupoid
$\tck{X_{n-1}}$. A \emph{cubical $(\omega,0)$-polygraph} is obtained
by a colimit construction; it consists of a family of sets
$X=(X_0,X_1,\dots)$ such that every subfamily $(X_0,\dots,X_n)$ is a
cubical $(n,0)$-polygraph. A polygraph is acyclic if and only if the
associated free groupoid is acyclic.

All polygraphs considered in the sequel are cubical. 

In order to construct acyclic polygraphs that extend a convergent
ARS in Sections~\ref{SS:PolResolFromConfl}
and~\ref{SS:PolResolFromConflTrunc}, we characterise acyclicity via
the existence of contractions in
Theorem~\ref{T:AcyclicityNormalisation}, adapting a similar result for globular
polygraphs~\cite{GuiraudMalbos12advances}. We start with the following
characterisation of contractions.
\begin{lemma}
\label{L:GenerationNormalisationStrategy}
Let $X$ be an $(\omega,0)$-polygraph and $\widehat{(-)}$ a section of
the projection $\pi:\tck{X}\to\overline{\tck{X}}_0$.  The contractions
of $\tck{X}$ are in bijective correspondence with the following data:
\begin{enumerate}
\item a family of $1$-cells $\sigma_x$ in $\tck{X}_1$ with boundary
  $x^\partial=(x,\widehat{x})$, for every $0$-cell $x$ in $X_0$ such
  that $x\neq\widehat{x}$,
\item a family of $(k+1)$-cells $\sigma_f$ in $\tck{X}_{k+1}$, for
  every $k>0$, with boundary $f^\partial$, for every $k$-cell $f$ in
  $X_k$ that is not of the form $\sigma_g$ for some $g$ in
  $\tck{X}_{k-1}$.
\end{enumerate}
Here, $f^\partial$ is defined recursively with respect to the
dimension of $k$-cells $f$ in $X_k^\top$, as in
Section~\ref{SSS:ContractionsCub}.
\end{lemma}
\begin{proof}
  A contraction has fixed values on thin cells, $R$-inverses,
  compositions and on elements of the form $\widehat{x}$ for
  $x\in X_0$ or $\sigma_g$ for some $g$ in $\tck{X}_{k-1}$. So the
  values of $\sigma_f$ for $f$ in $\tck{X}_k$ are uniquely and
  completely determined by its values on generators of the form given
  in the lemma. A construction of the free groupoid $\tck{X}$ can be
  found in Appendix~\ref{A:CubPolFreeCat}.
\end{proof}

We can now prove the converse direction to
Theorem~\ref{T:ContractingImpliesAcyclic} for
polygraphs.

\begin{theorem}
\label{T:AcyclicityNormalisation}
The free $\omega$-groupoid generated by an $(\omega,0)$-polygraph is
acyclic if and only if it is contracting.
\end{theorem}
\begin{proof}
  Let $X$ be an acyclic cubical $(\omega,0)$-polygraph.  We construct
  a contraction $\sigma$ recursively in the dimension of cells. This
  yields a contraction of the cubical $(k+1,0)$-polygraph
  $(X_0, \dots, X_{k+1})$ for each $0\leq k<n$.  For $k=0$ and every
  $0$-cell $x\in X_0$ such that $x\neq\widehat{x}$, we choose
  $\sigma_x:x\to\widehat{x}$ in $\tck{X_1}$, which exists by
  definition. This yields a contraction of $(X_0, X_1)$.  For $k>0$,
  suppose $\sigma$ is a contraction of $(X_0, \dots, X_k)$ and take a
  $k$-cell $f\in X_k$ which is not of the form $\sigma_g$ for some
  $g\in\tck{X_{k-1}}$. By acyclicity, the $k$-square $f^\partial$
  admits a filler $A$ in $\tck{X_{k+1}}$, and we set $\sigma_f :=
  A$. By Lemma~\ref{L:GenerationNormalisationStrategy}, $\sigma$
  extends to a contraction of the $(k+1,0)$-polygraph
  $(X_0, \dots, X_{k+1})$.  Taking the colimit yields a contraction of
  $X$.

  The reverse implication follows from
  Theorem~\ref{T:ContractingImpliesAcyclic}, considering the
  $\omega$-groupoid $\Ccal=\tck{X}$ freely generated by $X$.
\end{proof}

We use Theorem~\ref{T:AcyclicityNormalisation} in
Theorems~\ref{T:AcyclicExtensionARS}
and~\ref{T:TruncatedAcyclicExtensionARS} below to calculate acyclic extensions of ARS.

\subsection{An acyclic $\omega$-groupoid from convergence}
\label{SS:PolResolFromConfl}

Next we describe the construction that extends a convergent ARS
into an acyclic $\omega$-groupoid generated by its higher-order
branchings.

\subsubsection{Abstract rewriting systems}
For a $1$-polygraph $X$, we consider the cellular extension $X_1$ as
an ARS on $0$-cells in the free category $X^\ast$, as defined in
\SSS{SSS:RewritingSystemInCubCat}, and a section $\widehat{(-)}$
defined by the normal forms in~\SSS{SSS:ContractionsCub}.  A
\emph{normalisation strategy} for $X$ is a contraction
$\sigma : X_0 \fl X_1^\ast$ with respect to $\widehat{(-)}$ defined,
for each $x\in X_0$, as
\[
\sigma_x \: = \: \eta_x\circ_{1}\sigma_{\partial_{1}^+(\eta_x)},
\]
where $\eta_x\in X_1$ is the first rewriting step of $\sigma_x$.

\subsubsection{The polygraph $\cb{}{\omega}{X}$}
\label{SSS:PolygraphC(X)}
Let $X$ be a convergent $1$-polygraph and $\sigma$ a
normalisation strategy for~$X$.  For every $x\in X_0$, we fix a strict
order $<$ on the set $\{ f \in X_1 \mid \partial_{1}^- f = x \}$,
making $\eta_x$ the least element.  We construct an
$(\omega,0)$-polygraph involving higher-order branchings and their
confluences by transfinite recursion, defining a sequence of cellular
extensions $(\cb{}{k}{X})_{k\geq 0}$ by
\begin{enumerate}
\item $\cb{}{0}{X} := X_0$ and $\cb{}{1}{X} := X_1$,
\item For $k=2$,  $\cb{}{2}{X} := \{ \filler{2}{f_1, f_2} \mid f_1,f_2\in X_1^\ast,\; f_1<f_2, \; \partial^-_1(f_1)=\partial^-_1(f_2) \}$, whose face maps of the $2$-cell
$\filler{2}{f_1,f_2}$, drawn in~\eqref{E:shapeA2}, are defined by, for $1\leq i\leq 2$,
\item For $k \geq 3$, 
$\cb{}{k}{X} := \{ \filler{k}{f_1, \dots, f_k} \mid f_i\in X_1^\ast,\; f_i<f_{i+1}, \; \partial^-_1(f_i)=\partial^-_1(f_{i+1})\;\text{for $1\leq i \leq k-1$} \}$, whose face maps of $\filler{k}{f_1, \dots, f_k}$ are defined by, for $1\leq i\leq k$,
\begin{eqn}{align*}
\partial_{i}^-\filler{k}{f_1,\dots,f_k}=\filler{k-1}{f_1,\dots,f_{i-1},f_{i+1},\dots,f_k}, \qquad
\partial_{i}^+\filler{k}{f_1,\dots,f_k}=\Gamma_{k-2}^-\dots\Gamma_{1}^-\sigma_{\partial_{1}^+(f_i)}.
\end{eqn}
\end{enumerate}
The $(\omega,0)$-polygraph
$\cb{}{\omega}{X}$ is the colimit of this construction.

The following lemma shows that the $k$-generators
$\filler{k}{f_1, \dots, f_k}$ are well-defined.
\begin{lemma}
For every $k\geq 2$, $C_k(X)$ defined in \SSS{SSS:PolygraphC(X)} is a cubical cellular extension of $\tck{\cb{}{k-1}{X}}$.
\end{lemma}
\begin{proof}
We need to check the square equations~\eqref{E:SquareEquations}. 
For $k=2$,
\begin{align*}
\partial_1^- \partial_1^- \filler{2}{f_1, f_2} &= x 
   = \partial_1^- \partial_2^- \filler{2}{f_1, f_2},\\
\partial_1^- \partial_1^+ \filler{2}{f_1, f_2} &= y_2 
   = \partial_1^+ \partial_2^- \filler{2}{f_1, f_2}, \\
\partial_1^+ \partial_1^- \filler{2}{f_1, f_2} &= y_1 
   = \partial_1^- \partial_2^+ \filler{2}{f_1, f_2},\\
\partial_1^+ \partial_1^+ \filler{2}{f_1, f_2} &= \widehat{x} 
   = \partial_1^+ \partial_2^+ \filler{2}{f_1, f_2},
\end{align*}
 which shows that $\filler{2}{f_1, f_2}$ forms a $2$-cell. For $k \geq 3$ and
 $1 \leq i < j \leq k$,
\begin{align*}
\partial_i^- \partial_j^- \filler{k}{f_1, \dots, f_k} &= \filler{k-2}{f_1, \dots, f_{i-1}, f_{i+1}, \dots, f_{j-1}, f_{j+1}, \dots, f_k} = \partial_{j-1}^- \partial_i^- \filler{k}{f_1, \dots, f_k}, \\
\partial_i^- \partial_j^+ \filler{k}{f_1, \dots, f_k} &= \Gamma_{k-3}^- \dots \Gamma_1^- \sigma_{t_0(f_j)} = \partial_{j-1}^+ \partial_i^- \filler{k}{f_1, \dots, f_k}, \\
\partial_i^+ \partial_j^- \filler{k}{f_1, \dots, f_k} &= \Gamma_{k-3}^- \dots \Gamma_1^- \sigma_{t_0(f_i)} = \partial_{j-1}^- \partial_i^+ \filler{k}{f_1, \dots, f_k}, \\
\partial_i^+ \partial_j^+ \filler{k}{f_1, \dots, f_k} &= \Gamma_{k-3}^- \dots \Gamma_i^- \partial_i^+ \Gamma_i^- \dots \Gamma_1^- \sigma_{t_0(f_j)} \\
&= \Gamma_{k-3}^- \dots \Gamma_i^- \epsilon_i \dots \epsilon_1 \partial_1^+ \sigma_{t_0(f_j)} 
= \epsilon_1 \dots \epsilon_1 \widehat{x} \\
&= \Gamma_{k-3}^- \dots \Gamma_{j-1}^- \epsilon_{j-1} \dots \epsilon_1 \partial_1^+ \sigma_{t_0(f_i)} \\
&= \Gamma_{k-3}^- \dots \Gamma_{j-1}^- \partial_{j-1}^+ \Gamma_{j-1}^- \dots \Gamma_1^- \sigma_{t_0(f_i)} \\
&= \partial_{j-1}^+ \partial_i^+ \filler{k}{f_1, \dots, f_k}.
\end{align*}
Thus $\filler{k}{f_1, \dots, f_k}$ is a $(k-1)$-square.
\end{proof}

\subsubsection{Extending $\sigma$ to a contraction of $\tck{\cb{}{\omega}{X}}$}
\label{SSS:ExtendingSigma}
We further extend $\sigma$ to a contraction of the $\omega$-groupoid
$\tck{\cb{}{\omega}{X}}$.  By
Lemma~\ref{L:GenerationNormalisationStrategy}, it suffices to define a
$(k+1)$-cell $\sigma_f$, for each $k\geq1$, only for those
$k$-generators $f$ that are not of the form $\sigma_g$ for some
$g\in\tck{\cb{}{k-1}{X}}$.

For each $f\in X_1$ not of the form $\sigma_z$ for some
$z\in X_0$, we define the following $2$-cell $\sigma_f$ in
$\tck{\cb{}{\omega}{X}}$ that fills the $1$-square:
\begin{equation*}
f^\partial = 
\begin{tikzcd}[global scale = 2 and 2 and 1 and 1.2]
x \ar[rr, "f"] \ar[dd, "\sigma_x"'] && y \ar[dd, "\sigma_y"] \\\\
\widehat{x} \ar[rr, equal] && \widehat{x}
\end{tikzcd}
\end{equation*}
If $f\neq\eta_x$ with $x=\partial_{1}^-(f)$, and
$x'=\partial_{1}^+(\eta_x)$, then we set
\begin{eqn}{equation*}
\sigma_f := \filler{2}{\eta_x,f}\circ_{1}\Gamma_{1}^-\sigma_{x'}=
\begin{tikzcd}[global scale = 2 and 2 and 1 and 1.2]
x \ar[rr, "f"] \ar[dd, "\eta_x"'] && y \ar[dd, "\sigma_y"] \\\\
x' \ar[rr, "\sigma_{x'}"'] \ar[dd, "\sigma_{x'}"'] && \widehat{x} \ar[dd, equal] \\\\
\widehat{x} \ar[rr, equal] && \widehat{x}
\end{tikzcd}
\end{eqn}
where $\sigma_x=\eta_x\circ_{1}\sigma_{x'}$ and $\eta_x\in X_1$.
Otherwise, if $f=\eta_x$, we set
\begin{eqn}{equation*}
\sigma_{\eta_x} := \Gamma_{1}^-\eta_x \circ_{1} \epsilon_{2}\sigma_{x'} =
\begin{tikzcd}[global scale = 2 and 2 and 1 and 1.2]
x \ar[rr, "\eta_x"] \ar[dd, "\eta_x"'] && x' \ar[dd, equal] \\\\
x' \ar[rr, equal] \ar[dd, "\sigma_{x'}"'] && x' \ar[dd, "\sigma_{x'}"] \\\\
\widehat{x} \ar[rr, equal] && \widehat{x}
\end{tikzcd}
\end{eqn}

For $k\geq 2$, let $A$ be a $k$-generator in $\cb{}{k}{X}$ that is not
of the form $\sigma_g$ for some $g\in\tck{\cb{}{k-1}{X}}$. Then
$A=\filler{k}{f_1,\dots,f_k}$ and we define a $(k+1)$-cell $\sigma_A$ in
$\tck{\cb{}{\omega}{X}}$ that fills the $k$-square $A^\partial$. If
$f_i\neq\eta_x$ for all $i$, where $x=\partial_{1}^-(f_i)$ and
$x'=\partial_{1}^+(\eta_x)$, then we set
\begin{eqn}{equation}
\label{E:NormStratDimN}
\sigma_A :=\filler{k+1}{\eta_x,f_1,\dots,f_k}\circ_{1} \Gamma_{k}^- \dots \Gamma_{1}^- \sigma_{x'}.
\end{eqn}
If $f_1=\eta_x$, then we set
\begin{eqn}{equation}
\label{E:NormStratDimNEta}
\sigma_A := \Gamma_{1}^- A \circ_{1} \epsilon_{2} \Gamma_{k-1}^- \dots \Gamma_{1}^- \sigma_{x'}.
\end{eqn}

\begin{lemma}
Each $\sigma_A$ defined as above is well-defined and a filler of $A^\partial$.
\end{lemma}
\begin{proof}
In the case (\ref{E:NormStratDimN}), the formula is well-defined
because
\begin{equation*}
\partial_1^+ \filler{k+1}{\eta_x, f_1, \dots, f_k} =  \Gamma_{k-1}^- \dots \Gamma_1^- \sigma_{x'}
 =  \partial_1^- \Gamma_k^- \dots \Gamma_1^- \sigma_{x'}.
\end{equation*}
Also, $\sigma_A$ fills $A^\partial$ because,
\begin{align*}
\partial_1^- \sigma_A
&= \partial_1^- \filler{k+1}{\eta_x, f_1, \dots, f_k}
= \filler{k}{f_1, \dots, f_k}
= A,
\\
\partial_1^+ \sigma_A
&= \partial_1^+ \Gamma_1^- \dots \Gamma_1^- \sigma_{x'}
= \epsilon_1 \dots \epsilon_1 \partial_1^+ \sigma_{x'}
= \epsilon_{k+1} \dots \epsilon_1 \widehat{x},
\end{align*}
and, for $j>1$,
\begin{align*}
\partial_j^- \sigma_A &= \partial_j^- \filler{k+1}{\eta_x, f_1, \dots, f_k} \circ_1 \partial_j^- \Gamma_k^- \dots \Gamma_1^- \sigma_{x'} \\
&= \filler{k}{\eta_x, f_1, \dots, f_{j-2}, f_j, \dots, f_k} \circ_1 \Gamma_{k-1}^- \dots \Gamma_j^- \partial_j^- \Gamma_j^- \dots \Gamma_1^- \sigma_{x'} \\
&= \filler{k}{\eta_x, f_1, \dots, f_{j-2}, f_j, \dots, f_k} \circ_1 \Gamma_{k-1}^- \dots \Gamma_1^- \sigma_{x'} \\
&= \sigma_{\filler{k-1}{f_1, \dots, f_{j-2}, f_j, \dots, f_k}} = \sigma_{\partial_{j-1}^- A},
\end{align*}
\begin{align*}
\partial_j^+ \sigma_A &= \partial_j^+ \filler{k+1}{\eta_x, f_1, \dots, f_k} \circ_1 \partial_j^+ \Gamma_k^- \dots \Gamma_1^- \sigma_{x'} \\
&= \Gamma_{k-1}^- \dots \Gamma_1^- \sigma_{t_0(f_{i-1})} \circ_1 \Gamma_{k-1}^- \dots \Gamma_j^- \epsilon_j \partial_j^+ \Gamma_{j-1}^- \dots \Gamma_1^- \sigma_{x'} \\
&= \Gamma_{k-1}^- \dots \Gamma_1^- \sigma_{t_0(f_{i-1})} \circ_1 \epsilon_{k-1} \dots \epsilon_j \partial_j^+ \Gamma_{j-1}^- \dots \Gamma_1^- \sigma_{x'} \\
&= \Gamma_{k-1}^- \dots \Gamma_1^- \sigma_{t_0(f_{i-1})} \circ_1 \epsilon_{k-1} \dots \epsilon_1 \partial_1^+ \sigma_{x'} \\
&= \Gamma_{k-1}^- \dots \Gamma_1^- \sigma_{t_0(f_{i-1})} = \Gamma_{k-1}^- \dots \Gamma_2^- \sigma_{\sigma_{t_0(f_{i-1})}} \\
&= \sigma_{\Gamma_{k-2}^- \dots \Gamma_1^- \sigma_{t_0(f_{i-1})}} = \sigma_{\partial_{j-1}^+ A}.
\end{align*}

In the case (\ref{E:NormStratDimNEta}), the formula is well-defined because
\begin{equation*}
\partial_1^+ \Gamma_1^- A = \epsilon_1 \partial_1^+ A
 = \epsilon_1 \Gamma_{k-2}^- \dots \Gamma_1^- \sigma_{x'}
 = \partial_1^- \epsilon_2 \Gamma_{k-1}^- \dots \Gamma_1^- \sigma_{x'},
\end{equation*}
and $\sigma_A$ fills $A^\partial$ because, we have
$\partial_1^- \sigma_A = \partial_1^- \Gamma_1^- A = A$, and
\begin{align*}
\partial_1^+ \sigma_A
= \partial_1^+ \epsilon_2 \Gamma_{k-1}^- \dots \Gamma_1^- \sigma_{x'}
= \epsilon_1 \partial_1^+ \Gamma_1^- \dots \Gamma_1^- \sigma_{x'}
= \epsilon_1 \dots \epsilon_1 \partial_1^+ \sigma_{x'}
= \epsilon_{k+1} \dots \epsilon_1 \widehat{x},
\end{align*}
and, for $j>1$, 
\begin{align*}
\partial_j^- \sigma_A &= \partial_j^- \Gamma_1^- A \circ_1 \partial_j^- \epsilon_2 \Gamma_{k-1}^- \dots \Gamma_1^- \sigma_{x'} \\
&= \Gamma_1^- \partial_{j-1}^- A \circ_1 \epsilon_2 \Gamma_{k-2}^- \dots \Gamma_{j-1}^- \partial_{j-1}^- \Gamma_{j-1}^- \dots \Gamma_1^- \sigma_{x'} \\
&= \Gamma_1^- \filler{k-1}{f_1, \dots, f_{j-2}, f_j, \dots, f_k} \circ_1 \epsilon_2 \Gamma_{k-2}^- \dots \Gamma_1^- \sigma_{x'} = \sigma_{\partial_{j-1}^- A},\\
&\\
\partial_j^+ \sigma_A &= \partial_j^+ \Gamma_1^- A \circ_1 \partial_j^+ \epsilon_2 \Gamma_{k-1}^- \dots \Gamma_1^- \sigma_{x'} \\
&= \Gamma_1^- \partial_{j-1}^+ A \circ_1 \epsilon_2 \Gamma_{k-2}^- \dots \Gamma_{j-1}^- \epsilon_{j-1} \partial_{j-1}^+ \Gamma_{j-2}^- \dots \Gamma_1^- \sigma_{x'} \\
&= \Gamma_1^- \Gamma_{k-2}^- \dots \Gamma_1^- \sigma_{t_0(f_{j-1})} \circ_1 \epsilon_2 \epsilon_{k-2} \dots \epsilon_{j-1} \partial_{j-1}^+ \Gamma_{j-2}^- \dots \Gamma_1^- \sigma_{x'} \\
&= \Gamma_{k-1}^- \dots \Gamma_1^- \sigma_{t_0(f_{j-1})} \circ_1 \epsilon_2 \epsilon_{k-2} \dots \epsilon_1 \partial_1^+ \sigma_{x'} \\
&= \Gamma_{k-1}^- \dots \Gamma_1^- \sigma_{t_0(f_{j-1})} \circ_1 \epsilon_1 \dots \epsilon_1 \widehat{x} \\
&= \Gamma_{k-1}^- \dots \Gamma_1^- \sigma_{t_0(f_{j-1})} 
= \Gamma_{k-1}^- \dots \Gamma_2^- \sigma_{\sigma_{t_0(f_{j-1})}} \\
&= \sigma_{\Gamma_{k-2}^- \dots \Gamma_1^- \sigma_{t_0(f_{j-1})}} 
= \sigma_{\partial_{j-1}^+ A}.\qedhere
\end{align*}
\end{proof}

We can now state the main theorem of this section.
\begin{theorem}
\label{T:AcyclicExtensionARS}
Every convergent $1$-polygraph $X$ extends to an acyclic
$\omega$-groupoid $\tck{\cb{}{\omega}{X}}$.
\end{theorem}
\begin{proof}
In ~\SSS{SSS:ExtendingSigma}, we have defined a family of
$1$-cells $\sigma_x$ in $\tck{X}_1$ with boundary
$x^\partial=(x,\widehat{x})$, for every $0$-cell $x$ in $X_0$ such
that $x\neq\widehat{x}$. We have also defined
a family of $(k+1)$-cells $\sigma_f$ in $\tck{\cb{}{k+1}{X}}$, for every $0<k<n$,  with boundary $f^\partial$, for every $k$-cell $f$ in $\cb{}{k}{X}$ which is not of the form $\sigma_g$ for some $g$ in $\tck{\cb{}{k-1}{X}}$.
Then $\sigma$ is
a contraction on $\cb{}{\omega}{X}$ by
Lemma~\ref{L:GenerationNormalisationStrategy} and the claim follows from
Theorem~\ref{T:AcyclicityNormalisation}.
\end{proof}

\subsection{A refined acyclic $\omega$-groupoid from convergence}
\label{SS:PolResolFromConflTrunc}

Finally, we refine the construction leading to
Theorem~\label{T:AcyclicExtensionARS} so that it generates an acyclic
$\omega$-groupoid from a ARS without introducing any generating cells
of dimension higher than $2$.  We begin with a technical lemma, which
is an immediate consequence of~\cite[Prop.~2.1]{Higgins2005}.

\begin{lemma}
\label{L:SquaresThinFaces}
In every $\omega$-groupoid, each $k$-square with thin faces can be filled by a thin cell.
\end{lemma}
\begin{proof}
  Let $S$ be a $k$-square. Applying the folding maps, as defined in
  \SSS{SSS:FoldingUnfolding}, yields a $k$-square $T=\Phi_k(S)$, which
  satisfies $\partial_1^-\Psi_kT=\partial_1^+\Psi_kT$
  by~\cite[Prop. 3.6]{AlAglBrownSteiner2002} and has a unique thin
  filler $B$ by~\cite[Prop. 2.1(iii)]{Higgins2005}. Applying the
  unfolding maps yields a $k$-cell $A=\overline{\Phi}_k(S,B)$ which is
  a filler of $S$ by Lemma~\ref{L:UnFoldingProperties}.
\end{proof}

\begin{theorem}
\label{T:TruncatedAcyclicExtensionARS}
Every convergent $1$-polygraph $X$ extends to an acyclic
$\omega$-groupoid $\tck{\cb{tr}{\omega}{X}}$ which is  generated by the $(\omega,0)$-polygraph defined by
\begin{align*}
\cb{tr}{0}{X} := X_0,
\qquad
\cb{tr}{1}{X} := X_1,
\qquad
\cb{tr}{2}{X} := 
\{ \filler{2}{\eta_x,f} \mid f\in X_1, \; \partial_{1}^+(f)=x, \; \eta_x\neq f \},
\end{align*}
where the boundary of $\filler{2}{\eta_x,f}$ is given
by~\eqref{E:shapeA2}, and which has no $k$-generators for $k>2$.
\end{theorem}
\begin{proof}
  Let $X$ be a convergent $1$-polygraph equipped with the normal form
  section and with normalisation strategy $\sigma$.  We consider the
  acyclic $\omega$-groupoid $\tck{\cb{}{\omega}{X}}$ from
  Theorem~\ref{T:AcyclicExtensionARS}. Let $(f_1,f_2)$ be a local
  branching with source $x$ such that $f_1,f_2\neq\eta_x$, let $x'$ be
  the target of $\eta_x$. The $2$-generator $\filler{2}{f_1,f_2}$ has
  the same faces as the $2$-cell
\begin{eqn}{equation}
\label{E:HomReduc1}
\left(\Gamma_1^+\eta_x \circ_2 \filler{2}{\eta_x,f_2}\right) \circ_1 \left(T_1\filler{2}{\eta_x,f_1}\circ_2 \Gamma_1^-\sigma_{x'}\right) =
\begin{tikzcd}[global scale = 2.6 and 2.1 and 1 and 1.2]
x \ar[rr, equal] \ar[dd, equal] && x \ar[rr, "f_2"] \ar[dd] && y_2 \ar[dd, "\sigma_{y_2}"] \\
& \Gamma_1^+\eta_x && \filler{2}{\eta_x,f_2} & \\
x \ar[rr] \ar[dd, "f_1"'] && x' \ar[rr] \ar[dd] && \widehat{x} \ar[dd, equal] \\
& T_1\filler{2}{\eta_x,f_1} && \Gamma_1^-\sigma_{x'} & \\
y_1 \ar[rr, "\sigma_{y_1}"'] && \widehat{x} \ar[rr, equal] && \widehat{x}
\end{tikzcd}
\end{eqn}
We can thus replace $\filler{2}{f_1,f_2}$ by this $2$-cell that
  depends only on the generators $\filler{2}{\eta_x,f_1}$,
  $\filler{2}{\eta_x,f_2}$.

  Let $(f_1,f_2,f_3)$ be a local $3$-branching with source $x$ and let
  $x'$ be the target of $\eta_x$. Suppose $f_1=\eta_x$ and
  $f_2,f_3\neq\eta_x$. The $3$-generator $\filler{3}{\eta_x,f_2,f_3}$
  has faces $\filler{2}{\eta_x,f_2}$, $\filler{2}{\eta_x,f_3}$,
  $\filler{2}{f_2,f_3}$ and three thin cells. We replace
  $\filler{2}{f_2,f_3}$ by~\eqref{E:HomReduc1}, so that
  $\filler{3}{\eta_x,f_2,f_3}$ has the same faces as the $3$-cell
\begin{eqn}{equation}
\label{E:HomReduc2}
\left(\Gamma_1^-\Gamma_1^+\eta_x\circ_3\Gamma_1^-\filler{2}{\eta_x,f_3}\right)\circ_2\left(\Gamma_2^-T_1\filler{2}{\eta_x,f_2}\circ_3\Gamma_2^-\Gamma_1^-\sigma_{x'}\right).
\end{eqn}
The cases where $f_2=\eta_x$ or $f_3=\eta_x$ lead to similar thin
cells.

Now suppose $f_1,f_2,f_3\neq\eta_x$. If we replace
$\filler{2}{f_1,f_2}$, $\filler{2}{f_1,f_3}$ and $\filler{2}{f_2,f_3}$
by~\eqref{E:HomReduc1}, then the $3$-generator
$\filler{3}{f_1,f_2,f_3}$ has the same faces as the $3$-cell
\begin{eqn}{gather}
\label{E:HomReduc3}
\left(\left(\left(\Gamma_1^-\Gamma_1^+\eta_x\circ_1R_1\Gamma_1^-\Gamma_1^-\eta_x\right)\circ_3\Gamma_1^-\filler{2}{\eta_x,f_3}\right)
\circ_2\left(T_2\Gamma_1^-T_1\filler{2}{\eta_x,f_2}\circ_3\Gamma_2^-\Gamma_1^-\sigma_{x'}\right)\right) \\
\hspace{9cm} \circ_1\,\Gamma_2^-\left(T_1\filler{2}{\eta_x,f_1}\circ_2\Gamma_1^-\sigma_{x'}\right). \notag
\end{eqn}
So again we replace $\filler{3}{f_1,f_2,f_3}$ by this thin cell.

Lemma~\ref{L:SquaresThinFaces} implies that, if we replace the faces
of any $4$-generator in $\tck{\cb{}{\omega}{X}}$ by the thin $3$-cells
described in formulas~\eqref{E:HomReduc2} and~\eqref{E:HomReduc3},
then the $4$-generator itself can be replaced by a thin cell.  The
same argument applies inductively in all higher dimensions.

This allows constructing a trucacted $(\omega,0)$-polygraph
$\cb{tr}{\omega}{X}$ from the acyclic $\omega$-groupoid
$\tck{\cb{}{\omega}{X}}$, retaining only the $0$-generators, the
$1$-generators and the $2$-generators of the form
$\filler{2}{\eta_x,f}$, where $f\in X_1$ and $f\neq\eta_x$. By
construction, it freely generates an acyclic $\omega$-groupoid
$\tck{\cb{tr}{\omega}{X}}$.  In particular, it has no $k$-generators
and no non-thin $k$-cells for any~$k\geq3$.
\end{proof}

\subsubsection{Example}
To illustrate the difference between Theorem~\ref{T:AcyclicExtensionARS} 
and Theorem~\ref{T:TruncatedAcyclicExtensionARS}, we consider the $1$-polygraph $X$ 
defined by the diagram
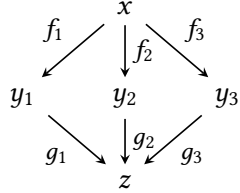
\begin{equation*}
\begin{tikzcd}[global scale = 3.5 and 3 and 1 and 1.2]
& x \ar[dl, "f_1"'] \ar[d, "f_2"] \ar[dr, "f_3"] & \\
y_1 \ar[dr, "g_1"'] & y_2 \ar[d, "g_2"] & y_3 \ar[dl, "g_3"] \\
& z &
\end{tikzcd}
\end{equation*}
It is convergent, and $z$ is the normal form of every $0$-cell. We
define the normalisation strategy $\sigma$ by
$\sigma_x=f_1\circ_{1}g_1$, $\sigma_{y_i}=g_i$ for every
$1\leq i\leq 3$, and $\sigma_z=1_z$, and set $\eta_x=f_1$  and $f_1<f_2<f_3$.

The ARS $X$ has the critical $2$-branchings $(f_1,f_2)$,
$(f_1,f_3)$, $(f_2,f_3)$ and the critical $3$-branching
$(f_1,f_2,f_3)$. The $(\omega,0)$-polygraph $\cb{}{\omega}{X}$
extending $X$ has the $2$-generators $\filler{2}{f_1,f_2}$,
$\filler{2}{f_1,f_3}$, $\filler{2}{f_2,f_3}$ and the $3$-generator
$\filler{3}{f_1,f_2,f_3}$.  The $\omega$-groupoid
$\tck{\cb{}{\omega}{X}}$ freely generated this way is acyclic.

By contrast, the $(2,0)$-polygraph $\cb{tr}{2}{X}$ extending $X$ has
the $2$-generators $\filler{2}{f_1,f_2}$, $\filler{2}{f_1,f_3}$, but
no $3$-generator. The $2$-groupoid $\tck{\cb{tr}{2}{X}}$ freely
generated this alternative way also acyclic. The critical
$2$-branching $(f_2,f_3)$, for instance, converges to $z$ via the
confluence $(g_2,g_3)$, and it gives rise to the $1$-square
$S=(f_2,f_3,g_3,g_2)$, filled with the $2$-cell
\begin{equation*}
(\Gamma_1^+f_1 \circ_2 \filler{2}{f_1,f_3}) \circ_1 (T_1\filler{2}{f_1,f_2}\circ_2 \Gamma_1^-g_1) =
\begin{tikzcd}[global scale = 2.6 and 2.1 and 1 and 1.2]
\ar[rr, equal] \ar[dd, equal] && \ar[rr, "f_3"] \ar[dd] && \ar[dd, "g_3"] \\
& \Gamma_1^+f_1 && \filler{2}{f_1,f_3} & \\
\ar[rr] \ar[dd, "f_2"'] && \ar[rr] \ar[dd] && \ar[dd, equal] \\
& T_1\filler{2}{f_1,f_2} && \Gamma_1^-g_1 & \\
\ar[rr, "g_2"'] && \ar[rr, equal] &&
\end{tikzcd}
\end{equation*}
The critical $3$-branching $(f_1,f_2,f_3)$ converges to $z$ via the
confluence $(g_1,g_2,g_3)$. This induces the $2$-square $S$
defined by
\begin{gather*}
\partial_1^+S = \Gamma_1^-g_1,
\qquad 
\partial_2^+S = \Gamma_1^-g_2,
\qquad 
\partial_3^+S = \Gamma_1^-g_3, \\[1ex]
\partial_1^-S = (\Gamma_1^+f_1 \circ_2 \filler{2}{f_1,f_3}) \circ_1
 (T_1\filler{2}{f_1,f_2} \circ_2 \Gamma_1^-g_1),
 \qquad
\partial_2^-S = \filler{2}{f_1,f_3},
\qquad 
\partial_3^-S = \filler{2}{f_1,f_2}.
\end{gather*}
It can be filled by the thin $3$-cell
\begin{equation*}
(\Gamma_1^-\Gamma_1^+f_1 \circ_3 \Gamma_1^-\filler{2}{f_1,f_3}) \circ_2 (\Gamma_2^-T_1\filler{2}{f_1,f_2} \circ_3 \Gamma_2^-\Gamma_1^-g_1).
\end{equation*}
Then $\tck{\cb{tr}{\omega}{X}}$ is indeed acyclic; the $2$-generator
$\filler{2}{f_2,f_3}$ and the $3$-generator $\filler{3}{f_1,f_2,f_3}$
are no longer needed.

\subsection{Concluding remarks}

The only $3$-confluence fillers in the proof of
Theorem~\ref{T:TruncatedAcyclicExtensionARS} are thin cells. The
$2$-confluence fillers employed are normalising, as explained in
Remark~\ref{SSS:ExSkolemNormConfl}, and the cube law holds \emph{a
  fortiori}. Hence, the cube law always holds for any ARS, since
rewriting rules have no application context and the critical branching
lemma from classical rewriting is trivial.

By contrast, in algebraic rewriting systems (string, term, linear,
etc.), the cube law is not inherent and must be proved separately
-- as, for instance, in the $\lambda$-calculus (see
\ref{SSS:ResiduationConflCubeLaw}). Future work will apply the cubical
constructions developed in this paper to such systems. Note also that,
unlike for ARS, convergent algebraic extensions generally do not
terminate after finitely many steps (see
Theorem~\ref{T:TruncatedAcyclicExtensionARS}).

In globular higher-dimensional rewriting, the constructions of
$\omega$-groupoids and related structures from polygraphs are known as
\emph{polygraphic resolutions}, as mentioned in the introduction, and
contractions may be regarded as contracting homotopies. This
topological terminology is justified by the folk model structure on
strict $\omega$-categories and the fact that polygraphic resolutions
are cofibrant
approximations~\cite{Metayer03,Metayer2008,LafontMetayerWorytkiewicz2010}.
In the cubical case, much less is known; polygraphic resolutions as
cofibrant approximations remain an avenue for future work. The proof
of Theorem~\ref{T:TruncatedAcyclicExtensionARS} has been inspired in
particular by a categorical approach to Tietze transformations in
globular polygraphs~\cite{GaussentGuiraudMalbos15}, which appears
worth exploring via cubical categories as well.

\begin{small}
\renewcommand{\refname}{\Large\textsc{References}}
\bibliographystyle{plain}
\bibliography{biblioFormalCubes}
\end{small}

\quad

\vfill

\begin{footnotesize}

\bigskip
\auteur{Philippe Malbos}{malbos@math.univ-lyon1.fr}
{Universit\'e Claude Bernard Lyon 1\\
ICJ UMR5208, CNRS\\
F-69622 Villeurbanne cedex, France}

\bigskip
\auteur{Tanguy Massacrier}{massacrier@math.univ-lyon1.fr}
{Universit\'e Claude Bernard Lyon 1\\
ICJ UMR5208, CNRS\\
F-69622 Villeurbanne cedex, France}

\bigskip
\auteur{Georg Struth}{g.struth@sheffield.ac.uk}
{University of Sheffield\\
  Department of Computer Science\\
  Regent Court, 211 Portobello\\
  Sheffield S1 4DP, United Kingdom
}
\end{footnotesize}

\vspace{1.5cm}

\begin{small}---\;\;\today\;\;-\;\;\hhmm\;\;---\end{small}

\clearpage

\appendix

\section{Appendices}

\subsection{Axioms of cubical categories}
\label{A:AxiomsCubCat}
We give a comprehensive axiomatisation of cubical categories, which were
outlined in Subsection~\ref{SS:CubicalCategories}.

\subsubsection{Cubical categories}
\label{AA:AxiomsCubCat}
Cubical categories satisfy the following axioms, for all $i,j,k\in\Nbb$ such that $1\leq i,j\leq k$:
\begin{align*}
\partial_{k,i}^\alpha\epsilon_{k,j}=
\begin{cases*}
\epsilon_{k-1,j-1}\partial_{k-1,i}^\alpha & if $i<j$, \\
\id_{C_{k-1}} & if $i=j$, \\
\epsilon_{k-1,j}\partial_{k-1,i-1}^\alpha & if $i>j$,
\end{cases*}
\end{align*}
\begin{equation*}
\epsilon_{k+1,i}\epsilon_{k,j+1} = \epsilon_{k+1,j}\epsilon_{k,i}
\quad \text{if $i\leq j$,}
\qquad
\epsilon_{k+1,i}\epsilon_{k,j} = \epsilon_{k+1,j}\epsilon_{k,i+1} \quad \text{if $i>j$,}
\end{equation*}
\begin{gather*}
(a\circ_{k,i}b)\circ_{k,j}(c\circ_{k,i}d)=(a\circ_{k,j}c)\circ_{k,i}(b\circ_{k,j}d), \\
a\circ_{k,i}(b\circ_{k,i}c)=(a\circ_{k,i}b)\circ_{k,i}c,
\end{gather*}
\begin{align*}
\epsilon_{k+1,i}(a\circ_{k,j}b)=
\begin{cases*}
\epsilon_{k+1,i}a\circ_{k+1,j+1}\epsilon_{k+1,i}b & if $i\leq j$, \\
\epsilon_{k+1,i}a\circ_{k+1,j}\epsilon_{k+1,i}b & if $i>j$,
\end{cases*}
\end{align*}
\begin{align*}
a\circ_{k,i}\epsilon_{k,i}\partial_{k,i}^+a=\epsilon_{k,i}\partial_{k,i}^-a\circ_{k,i}a=a,
\end{align*}
\begin{align*}
\partial_{k,i}^\alpha(a\circ_{k,j}b)=
\begin{cases*}
\partial_{k,i}^\alpha a\circ_{k,j-1}\partial_{k,i}^\alpha b & if $i<j$, \\
\partial_{k,i}^-a & if $i=j$ and $\alpha=-$, \\
\partial_{k,i}^+b & if $i=j$ and $\alpha=+$, \\
\partial_{k,i}^\alpha a\circ_{k,j}\partial_{k,i}^\alpha b & if $i>j$,
\end{cases*}
\end{align*}

\subsubsection{Connections}
\label{AA:AxiomsCubCatConnections}
Cubical categories with connections satisfy the following additional axioms:
\begin{align*}
\partial_{k,i}^\alpha\Gamma_{k,j}^\beta=
\begin{cases*}
\Gamma_{k-1,j-1}^\beta\partial_{k-1,i}^\alpha & if $i<j$, \\
\id_{C_{k-1}} & if $i=j,j+1$ and $\alpha=\beta$, \\
\epsilon_{k-1,j}\partial_{k-1,j}^\alpha & if $i=j,j+1$ and $\alpha=-\beta$, \\
\Gamma_{k-1,j}^\beta\partial_{k-1,i-1}^\alpha & if $i>j+1$,
\end{cases*}
\end{align*}
\begin{equation*}
\Gamma_{k+1,i}^\alpha\epsilon_{k,j} =
\begin{cases*}
\epsilon_{k+1,j+1}\Gamma_{k,i}^\alpha & if $i<j$, \\
\epsilon_{k+1,i}\epsilon_{k,i} & if $i=j$, \\
\epsilon_{k+1,j}\Gamma_{k,i-1}^\alpha & if $i>j$,
\end{cases*}
\qquad
\Gamma_{k+1,i}^\alpha\Gamma_{k,j}^\beta =
\begin{cases*}
\Gamma_{k+1,j+1}^\beta\Gamma_{k,i}^\alpha & if $i<j$, \\
\Gamma_{k+1,j}^\alpha\Gamma_{k,j}^\alpha & if $i=j+1$ and $\alpha=\beta$, \\
\Gamma_{k+1,j}^\beta\Gamma_{k,i-1}^\alpha & if $i>j+1$.
\end{cases*}
\end{equation*}
\begin{equation*}
\Gamma_{k,i}^+a\circ_{k,i}\Gamma_{k,i}^-a =\epsilon_{k,i+1}a,\qquad
\Gamma_{k,i}^+a\circ_{k,i+1}\Gamma_{k,i}^-a = \epsilon_{k,i}a,
\end{equation*}
\begin{align*}
\Gamma_{k+1,i}^\alpha(a\circ_{k,j}b)=
\begin{cases*}
\Gamma_{k+1,i}^\alpha a\circ_{k,j+1}\Gamma_{k+1,i}^\alpha b & if $i<j$, \\
(\Gamma_{k+1,i}^-a\circ_{k,i}\epsilon_{k+1,i+1}b)\circ_{k,i+1}(\epsilon_{k+1,i}b\circ_{k,i}\Gamma_{k+1,i}^-b) & if $i=j$ and $\alpha=-$, \\
(\Gamma_{k+1,i}^+a\circ_{k,i}\epsilon_{k+1,i}a)\circ_{k,i+1}(\epsilon_{k+1,i+1}a\circ_{k,i}\Gamma_{k+1,i}^+b) & if $i=j$ and $\alpha=+$, \\
\Gamma_{k+1,i}^\alpha a\circ_{k,j}\Gamma_{k+1,i}^\alpha b & if $i>j$,
\end{cases*}
\end{align*}

\subsubsection{Functors}
\label{AA:CubicalFunctors}
A functor $F:\Ccal\to\Dcal$ of cubical $\omega$-categories is a family of
maps $(F_k:\Ccal_k\to\Dcal_k)_{0\leq k}$ satisfying
\[
F_k(a \circ_{k,i} b) = F_k a \circ_{k,i} F_k b,
\qquad
F_{k-1}\partial_{k,i}^\alpha = \partial_{k,i}^\alpha F_k,
\qquad
F_k\epsilon_{k,i} = \epsilon_{k,i} F_{k-1},
\qquad
F_k\Gamma_{k,j}^\alpha = \Gamma_{k,j}^\alpha F_{k-1},
\]
for all $i,j,k\in\Nbb$ such that $1\leq i \leq k$ and $1\leq j<k$, and all $\circ_{k,i}$-composable $a,b\in\Ccal_k$.

\subsubsection{Inverses}
\label{AA:InversionMaps}
The inversion maps $R_{i}$ and $T_{i}$ defined in~\SSS{SS:DefGrpdCubCat} are compatible with
\begin{enumerate}
\item the face maps
\begin{equation*}
\partial_{i}^\alpha R_{j}f =
\begin{cases*}
R_{j-1}\partial_{i}^\alpha f & if $i<j$, \\
\partial_{i}^{-\alpha} f & if $i=j$, \\
R_{j}\partial_{i}^\alpha f & if $i>j$,
\end{cases*}
\qquad
\partial_{i}^\alpha T_{j}f =
\begin{cases*}
T_{j-1}\partial_{i}^\alpha f & if $i<j$, \\
\partial_{i+1}^\alpha f & if $i=j$, \\
\partial_{i-1}^\alpha f & if $i=j+1$, \\
T_{j}\partial_{i}^\alpha f & if $i>j+1$,
\end{cases*}
\end{equation*}
\item the compositions
\begin{equation*}
R_{i}(f\circ_{j}g) =
\begin{cases*}
R_{i}g\circ_{i}R_{i}f & if $i=j$, \\
R_{i}f\circ_{j}R_{i}g & if $i\neq j$,
\end{cases*}
\qquad
T_{i}(f\circ_{j}g) =
\begin{cases*}
T_{i}f\circ_{i+1}T_{i}g & if $j=i$, \\
T_{i}f\circ_{i}T_{i}g & if $j=i+1$, \\
T_{i}f\circ_{j}T_{i}g & if $j\neq i,i+1$,
\end{cases*}
\end{equation*}
\item the degeneracies
\begin{equation*}
R_{i}\epsilon_{j}f =
\begin{cases*}
\epsilon_{j}R_{i}f & if $i<j$, \\
\epsilon_{i}f & if $i=j$, \\
\epsilon_{j}R_{i-1}f & if $i>j$,
\end{cases*}
\qquad
T_{i}\epsilon_{j}f =
\begin{cases*}
\epsilon_{j}T_{i-1}f & if $j<i$, \\
\epsilon_{i+1}f & if $j=i$, \\
\epsilon_{i}f & if $j=i+1$, \\
\epsilon_{j}T_{i}f & if $j>i+1$,
\end{cases*}
\end{equation*}
\item the connections
\begin{gather*}
R_{i}\Gamma_{j}^\alpha f=
\begin{cases*}
\Gamma_{j}^\alpha R_{i}f & if $i<j$, \\
\epsilon_{i+1}R_{i}f\circ_{i+1}\Gamma_{i}^+f & if $i=j$, $\alpha=-$, \\
\Gamma_{i}^-f\circ_{i}\epsilon_{i+1}R_{i}f & if $i=j$, $\alpha=+$, \\
\epsilon_{i-1}R_{i-1}f\circ_{i}\Gamma_{i-1}^+f & if $i=j+1$, $\alpha=-$, \\
\Gamma_{i-1}^-f\circ_{i}\epsilon_{i-1}R_{i-1}f & if $i=j+1$, $\alpha=+$, \\
\Gamma_{j}^\alpha R_{i-1}f & if $i>j+1$,
\end{cases*}\
\qquad
T_{i}\Gamma_{j}^\alpha f =
\begin{cases*}
\Gamma_{j}^\alpha T_{i}f & if $i<j$, \\
\Gamma_{i}^\alpha f & if $i=j$, \\
\Gamma_{j}^\alpha T_{i-1}f & if $i>j$,
\end{cases*}\\[1em]
T_{i+1}\Gamma_{i}^\alpha T_{i}f = T_{i}\Gamma_{i+1}^\alpha f, \qquad
T_{i}\Gamma_{i+1}^\alpha T_{i}f = T_{i+1}\Gamma_{i}^\alpha f,
\end{gather*}
\item other inversion maps
\begin{gather*}
R_{i}R_{j}f =
\begin{cases*}
f & if $i=j$, \\
R_{j}R_{i}f & if $i\neq j$,
\end{cases*}
                  \qquad
                  T_{i}T_{j}f =
\begin{cases*}
f & if $i=j$, \\
T_{j}T_{i}f & if $|i-j|\geq2$,
\end{cases*}\\
T_{i}R_{j}f =
\begin{cases*}
R_{i+1}T_{i}f & if $j=i$, \\
R_{i}T_{i}f & if $j=i+1$, \\
R_{j}T_{i}f & if $j\neq i,i+1$.
\end{cases*}\qquad
T_{i}T_{i+1}T_{i}f = T_{i+1}T_{i}T_{i+1}f,
\end{gather*}
\end{enumerate}

\subsection{Cubical polygraphs and free cubical categories}
\label{A:CubPolFreeCat}

In this appendix, we detail the construction of the cubical polygraphs
used in Section~\ref{S:PolygraphicResolutionARS}.  Cubical polygraphs
form systems of generators for cubical categories, defined inductively
on the dimension.  Our presentation follows the method developed by
Métayer in the globular setting~\cite{Metayer2008}.  We first
introduce the notion of \emph{cubical extension}, a set of
$(n+1)$-generators adjoined to a cubical $n$-category.
Lemma~\ref{L:ConstructionFreeFunctorCubCat} makes the construction of
the free cubical $(n+1)$-category generated by a cubical $n$-category
and equipped with a cubical extension explicit.  This construction is
then used to define cubical polygraphs recursively by adjoining
cubical extensions to freely generated cubical categories.

\subsubsection{Cubical extensions}
For $n\in\mathbb{N}$, a \emph{precubical $n$-set} is a family
$\Ccal = (\Ccal_k)_{0\leq k\leq n}$ of $k$-cells with face maps
$\partial_{k,i}^\alpha : \Ccal_k \to \Ccal_{k-1}$, for
$1\leq i\leq k\leq n$, satisfying the cubical
relations~\eqref{E:AxiomPreCubClass}. A functor $F : \Ccal \fl \Dcal$
of precubical sets is a family of maps
$(F_k : \Ccal_k \fl \Dcal_k)_{k\in \mathbb{N}}$ that preserve face
maps, that is
$F_{k-1}\partial_{k,i}^\alpha = \partial_{k,i}^\alpha F_k$, for every
$1\leq i \leq n$. We denote by $\PreCub{n}$ the category of precubical
$n$-sets and their functors.  We denote by $\CubCatG{n}$ the category
of cubical $n$-categories and their functors as defined in
\SSS{SSS:DefCubCat}.

The \emph{category of cubical extensions of cubical $n$-categories} is defined by the following pullback in $\catego{CAT}$
\begin{equation*}
\begin{tikzcd}[global scale = 7 and 4 and 1 and 1]
(\CubCatG{n})^+ \ar[r, dotted] \ar[d, dotted] & \PreCub{n+1} \ar[d] \\
\CubCatG{n} \ar[r, "U_n"] & \PreCub{n}
\end{tikzcd}
\end{equation*}
where the bottom arrow is the forgetful functor and
the right arrow the truncation functor.

Explicitly, a cubical extension of a cubical $n$-category $\Ccal$
consists of a set $X_{n+1}$ of \emph{$(n+1)$-generators} and a set of
face maps $\partial_{n+1,i}^\alpha:X_{n+1}\to\Ccal_n$, for
$1\leq i\leq n+1$, that satisfy the cubical
relations~\eqref{E:AxiomPreCubClass}.  A morphism of cubical
extensions $F:(\Ccal,X)\to(\Dcal,Y)$ consists of a functor between the
cubical $n$-categories $G:\Ccal\to\Dcal$ and a map $H:X\to Y$ such
that $\partial_{n+1,i}^\alpha H=G_n\partial_{n+1,i}^\alpha$ for all
$1\leq i\leq n+1$.

Consider the forgetful functor
\[
W_n : \CubCatG{n+1}\to(\CubCatG{n})^+ 
\]
sending a cubical $(n+1)$-category $\Ccal$ to the pair
$(\Ccal_{\leq n},\Ccal_{n+1})$, where $\Ccal_{\leq n}$ is the
$n$-category made of $k$-cells of $\Ccal$, for $k\leq n$, and
$\Ccal_{n+1}$ is the set of $(n+1)$-cells viewed as a cubical
extension.  It has a left adjoint $L_{n}$, which maps a cubical
$n$-category $\Ccal$, equipped with a cubical extension $X_{n+1}$, to
the freely generated cubical $(n+1)$-category $\Ccal[X_{n+1}]$.  For
Gray categories and polygraphs, a proof of the existence of this
adjoint functor has been given by Lucas~\cite{LucasPhD2017}, although
no explicit construction is given there.  We provide a fully syntactic
construction of the free functor $L_n$ using a type system analogous
to that of Métayer in the globular case~\cite[Section
4.1]{Metayer2008}.  Our syntax differs from the globular one in
several respects: we introduce constants for degeneracy and connection
maps rather than identity maps, and we quotient by the cubical axioms
instead of the globular ones.  Another difference concerns the type of
$(n+1)$-cells: in the globular case one uses $n$-globes; here the
corresponding types are $n$-squares.

\begin{lemma}
\label{L:ConstructionFreeFunctorCubCat}
The forgetful functor $W_n:\CubCatG{n+1}\to(\CubCatG{n})^+$ has a left
adjoint $L_n$.
\end{lemma}
\begin{proof}
  Consider $(\Ccal,X_{n+1})$ in $(\CubCatG{n})^+$, with face maps of
  $X_{n+1}$ denoted $\partial_{n+1,i}^\alpha$ for all
  $1\leq i\leq n+1$. We define a formal syntax $\Er$ formed by
\begin{enumerate}
\item a constant symbol ${\bf c}_x$, for each $x\in X_{n+1}$,
\item a constant symbol ${\bf e}_{i,c}$, for each $c\in \Ccal_n$ and $1\leq i\leq n+1$,
\item a constant symbol ${\bf g}_{i,c}^\alpha$, for each $c\in\Ccal_n$
  and $1\leq i\leq n$,
\item a binary function symbol $\circ_i$, for each $1\leq i\leq n+1$.
\end{enumerate}
Then $\Er$ is the smallest set of that contains all constants and is
closed under the operation $A\circ_i B$, for all $A,B\in \Er$ and
$1\leq i\leq n+1$.  A \emph{type} is any $n$-square in $\Ccal_n$.  For
every $A\in \Er$ and every type $S$, we recursively defined the
judgement $A:S$ -- \emph{$A$ has type $S$}:
following rules:
\begin{enumerate}
\item ${\bf c}_x:\partial x$, for every $x\in X_{n+1}$,
\item ${\bf e}_{i,c}:S$, for every $n$-cell $c$ in $\Ccal$, where
\begin{align*}
  S_j^\alpha =
  \begin{cases}
     \epsilon_{n,i}\partial_{n,j-1}^\alpha c &\text{if } i<j,\\
     c &\text{if } i=j,\\
     \epsilon_{n,i-1}\partial_{n,j}^\alpha c  & \text{if } i>j,
  \end{cases}
 \end{align*}
\item ${\bf g}_{i,c}^\alpha:S$, for every $n$-cell $c$ in $\Ccal$, where
\begin{align*}
  S_j^\beta =
  \begin{cases}
   \Gamma_{n,i}^\alpha\partial_{n,j-1}^\beta c  &\text{if } i<j-1,\\
   c &\text{if } j=i,i+1 \text{ and } \alpha=\beta,\\
   \epsilon_{n,i}\partial_{n,i}^\alpha c &\text{if } j=i,i+1 \text{ and } \alpha=-\beta,\\
   \Gamma_{n,i-1}^\alpha\partial_{n,j}^\beta c  &\text{if } i>j,
  \end{cases}
\end{align*}
\item $(A\circ_i B):U$, for expressions $A:S$ and $B:T$, where
\begin{align*}
  U_j^\alpha=
  \begin{cases}
    S_j^\alpha\circ_{n,i}T_j^\alpha & \text{if } i<j,\\
    S_i^- &\text{if } i=j \text{ and } \alpha=-,\\
    T_i^+ &\text{if } i=j \text{ and } \alpha=+,\\
    S_j^\alpha\circ_{n,i-1}T_j^\alpha & \text{if } i>j.
  \end{cases}
 \end{align*}
\end{enumerate}

An expression $A$ is \emph{typable} if $A:S$ for some type $S$. A
simple structural induction shows that typable expressions are
uniquely type. Let $\Er_T\subseteq \Er$ denote the set of typable
expressions.  By uniqueness of types, there exist unique maps
$d_i^\alpha:\Er_T\to\Ccal_n$, for $1\leq i\leq n+1$, such that
$d_i^\alpha({\bf c}_x) =\partial_{n+1,i}^\alpha(x)$ and
$A:(d_i^\alpha(A))_{i,\alpha}$ for all $x\in\Ccal_n$ and $A\in \Er_T$.

We write $\rhd_i$ for the relation of being $\circ_i$-composable on
$\Ccal_n$. We extend this relation to $\Er_T$ by setting $A\rhd_i B$ if $d_i^-(A)=d_i^+(B)$.
Let $\sim$ be the smallest equivalence  on $\Er_T$ generated by the
following conditions, for all $1\leq i,j\leq n$, $A,B,C,D\in \Er_T$ and $c,d\in\Ccal_n$:
\begin{enumerate}
\item $A\circ_i (B\circ_i C)\sim(A\circ_i B)\circ_i C$, if $A\rhd_i B\rhd_i C$,
\item if $i<j$, $A\rhd_i B$, $C\rhd_i D$, $A\rhd_j C$ and $B\rhd_j D$, then
\begin{equation*}
(A\circ_i B) \circ_j (C\circ_i D) \sim (A\circ_j C) \circ_i (B\circ_j D),
\end{equation*}
\item ${\bf e}_{i,c}\circ_i A\sim A$, if $d_i^-(A)=c$,
\item $A\circ_i {\bf e}_{i,c}\sim A$, if $d_i^+(A)=c$,
\item if $c\rhd_i d$, then
\begin{gather*}
  {\bf e}_{i,c\circ_j d} \sim
  \begin{cases}
     {\bf e}_{i,c}\circ_{j+1} {\bf e}_{i,d} & \text{if }i\leq j, \\
     {\bf e}_{i,c}\circ_{j} {\bf e}_{i,d} & \text{if }i>j,
  \end{cases}
 \end{gather*}
\item ${\bf e}_{i,\epsilon_{n,j}c}\sim{\bf e}_{j+1,\epsilon_{n,i}c}$, if $i\leq j$,
\item if $c\rhd_i d$, then
\begin{equation*}
{\bf g}_{i,c\circ_{n,j}d}^\alpha \sim
\begin{cases}
{\bf g}_{i,c}^\alpha\circ_{j+1}{\bf g}_{i,d}^\alpha &\text{if }i<j, \\
({\bf g}_{i,c}^-\circ_i{\bf e}_{i+1,d})\circ_{i+1}({\bf e}_{i,d}\circ_i{\bf g}_{i,d}^-) &\text{if } i=j \text{ and } \alpha=-, \\
({\bf g}_{i,c}^+\circ_i{\bf e}_{i,c})\circ_{i+1}({\bf e}_{i+1,c}\circ_i{\bf g}_{i,d}^+) &\text{if } i=j \text{ and } \alpha=+, \\
{\bf g}_{i,c}^\alpha\circ_j{\bf g}_{i,d}^\alpha &\text{if }i>j, \\
\end{cases}
\end{equation*}
\item ${\bf g}_{i,c}^+\circ_i{\bf g}_{i,c}^-\sim{\bf e}_{i+1,c}$ and ${\bf g}_{i,c}^+\circ_{i+1}{\bf g}_{i,c}^-\sim{\bf e}_{i,c}$,
\item
\begin{gather*}
{\bf g}_{i,\epsilon_{n,j}c}^\alpha \sim
\begin{cases}
{\bf e}_{j+1,\Gamma_{n,i}^\alpha c} &\text{if } i<j, \\
{\bf g}_{i,\epsilon_{n,i}c}^\alpha \sim {\bf e}_{i,\epsilon_{n,i} c} &\text{if } i=j, \\
{\bf g}_{i,\epsilon_{n,j}c}^\alpha \sim {\bf e}_{j,\Gamma_{n,i-1}^\alpha c} &\text{if } i>j, \\
\end{cases}
\end{gather*}
\item ${\bf g}_{i,\Gamma_{n,j}^\beta c}^\alpha \sim {\bf g}_{j+1,\Gamma_{n,i}^\alpha c}^\beta$ if $i<j$ and ${\bf g}_{i,\Gamma_{n,i}^\alpha c}^\alpha \sim {\bf g}_{i+1,\Gamma_{n,i}^\alpha c}^\alpha$.
\end{enumerate}

Let $\cong$ be the congruence generated by $\sim$ on $\Er_T$.  We
define $\cat{X_{n+1}} \coloneq \Er_T/\cong$, and write~$[A]$ for the
equivalence class of an expression~$A$.  
we define the operations
\[
\partial_{n+1,i}^\alpha([A]) \:\coloneq\: d_i^\alpha(A)
\quad\text{and}\quad
[A_1]\circ_{n+1,i}[ A_2] \:\coloneq\: [ A_1\circ_i A_2],
\]
on $\cat{X_{n+1}}$ whenever $A_1\rhd_i A_2$.  We further define maps
$\epsilon_{n+1,i},\Gamma_{n+1,i}^\alpha:\Ccal_n \to \cat{X_{n+1}}$, for every $c\in \Ccal_n$, by
\[
\epsilon_{n+1,i}(c) \:\coloneq\: [ {\bf e}_{i,c}]
\quad\text{and}\quad
\Gamma_{n+1,i}^\alpha(c) \:\coloneq\: [ {\bf g}_{i,c}^\alpha],
\]

Finally, we define $L_n(\Ccal,X_{n+1})$ to be the cubical
$(n+1)$-category with underlying $n$-category $\Ccal_n$, set of
$(n+1)$-cells $\cat{X_{n+1}}$, and structure induced by the operations
just introduced.

It is routine to check that this construction produces a cubical
$(n+1)$-category, and that it extends to make
$L_n:(\CubCatG{n})^+\to\CubCatG{n+1}$ functorial.

Next, we check the adjunction $L_n\dashv W_n$.  Let $(\Ccal,X_{n+1})$
be in $(\CubCatG{n})^+$, let $\Dcal$ be in $\CubCatG{n+1}$ and let
\[
f \:\coloneq\: (g:\Ccal\to\Dcal_{\leq n},h:X_{n+1}\to\Dcal_{n+1})
\]
be a morphism $f:(\Ccal,X_{n+1})\to W_n(\Dcal)$ in
$(\CubCatG{n})^+$.

We recursively define a map $f':\Er_T \to \Dcal$, for all $x\in
X_{n+1}$, $c\in\Ccal_n$ and  $A,B\in \Er_T$,
$1\leq i\leq n+1$, as
\begin{equation*}
f'({\bf c}_x) = h(x), \qquad
f'({\bf e}_{i,c})  = \epsilon_{n+1,i}(g(c)),\qquad
f'({\bf g}_{i,c}^\alpha)  \Gamma_{n+1,i}^\alpha(g(c)),\qquad
f'(A\circ_i B) = f'(A)\circ_{n+1,i} f'(B).
\end{equation*}
It is compatible with $\cong$ in the sense that $f'(A)=f'(B)$ whenever
$A\cong B$, hence it induces a well-defined map
$\cat{f}:L_n(\Ccal,X_{n+1})\to\Dcal$. It is straightforward to check
that $\cat{f}$ is a cubical $(n+1)$-functor. Hence we obtain a map of
type
\begin{equation*}
(\CubCatG{n})^+((\Ccal,X_{n+1}),W_n(\Dcal))\to\CubCatG{n+1}(L_n(\Ccal,X_{n+1}),\Dcal).
\end{equation*} 
It is also easy to check that this map is natural in $(\Ccal,X_{n+1})$
and $\Dcal$, and that it is invertible, the inverse sending a cubical
$(n+1)$-functor $f$ to the pair $(g,h)$ where $g$ is the
$n$-truncation of $f$ and $h$ is the map between the sets of
$(n+1)$-cells. This yields a natural isomorphism between the above
hom-sets, which establishes $L_n\dashv W_n$.
\end{proof}

The construction of the left adjoint for cubical $(n,p)$-categories
proceeds as above, after adjoining inverse as constants to the syntax
and the associated invertibility axioms to the congruence~$\sim$. 

\begin{lemma}
\label{L:ConstructionFreeFunctorCubNPCat}
The forgetful functor $W_{(n,p)}:\GlobCat{(n+1,p)}\to\GlobCat{(n,p)}^+$ has a left
adjoint $L_{(n,p)}$.
\end{lemma}

\subsubsection{Cubical polygraphs}

We can now construct cubical polygraphs along the lines of their
globular siblings~\cite{Polybook2025}.  We recursively define the
categories $\CubPol{n}$ of \emph{cubical $n$-polygraphs} and the
functors $F_n:\CubPol{n}\to\CubCatG{n}$, which send a cubical
$n$-polygraph to the cubical $n$-category $F_n(X)=\cat{X}$ freely
generates by it:
\begin{enumerate}
\item The category $\CubPol{0}$ is $\catego{Set}$ and the functor $F_0$  the identity.
\item Given $\CubPol{n}$ and $F_n$, the category $\CubPol{n+1}$ is defined by the pullback 
\begin{equation*}
\begin{tikzcd}[global scale = 7 and 4 and 1 and 1]
\CubPol{n+1} \ar[r, dotted, "J_n"] \ar[d, dotted] & (\CubCatG{n})^+ \ar[d] \\
\CubPol{n} \ar[r, "F_n"] & \CubCatG{n}
\end{tikzcd}
\end{equation*}
in $\catego{CAT}$, and the functor $F_{n+1}$ is defined as the composition
\[
\CubPol{n+1}\oto{J_n}(\CubCatG{n})^+\oto{L_n}\CubCatG{n+1}.
\]
\end{enumerate}
Explicitly, a cubical $n$-polygraph is a family $(X_0,\dots,X_n)$, where each $X_{k+1}$ is a cubical extension of $\cat{X_{\leq k}}$ for every $k<n$.
The category $\CubPol{\omega}$ of \emph{cubical $\omega$-polygraphs}
is the projective limit of the following diagram in $\catego{CAT}$
\begin{equation*}
\CubPol{0}\overset{V_0}{\longleftarrow}\CubPol{1}\longleftarrow\dots
\longleftarrow\CubPol{n}\overset{V_n}{\longleftarrow}\CubPol{n+1}\longleftarrow\dots
\end{equation*}
where, for every $n\geq 1$, the functor $V_n$ is the truncation
functor forgetting the $(n+1)$-dimensional cubical extension.

Finally, adding inverses both to the definition of cubical polygraphs
and to the construction of the free cubical category in
Lemma~\ref{L:ConstructionFreeFunctorCubNPCat} leads to the notion of
\emph{cubical $(n,p)$-polygraphs} for all $p \leq n$.  Each cubical
$(n,p)$-polygraph $X$ freely generates a cubical $(n,p)$-category,
denoted~$\tck{X}$.

\end{document}